\documentclass[11pt,a4paper]{article}

\usepackage{xcolor}
\usepackage{framed}

\definecolor{shadecolor}{RGB}{240,240,240}
\usepackage[T1]{fontenc}
\usepackage{lmodern}
\usepackage[margin=1in]{geometry}
\usepackage{amsmath,amssymb,amsthm,mathtools}
\usepackage{microtype}
\usepackage{enumitem}
\usepackage{aliascnt}
\usepackage[ruled,linesnumbered,vlined]{algorithm2e}
\usepackage[hidelinks]{hyperref}
\usepackage[capitalize,noabbrev]{cleveref}

\hypersetup{
    pdftitle={Weighted EF1 Allocations for Additive Mixed Manna},
    pdfauthor={}
}

\allowdisplaybreaks
\setlist{itemsep=0.25em,topsep=0.45em}

\newtheorem{theorem}{Theorem}[section]

\newaliascnt{lemma}{theorem}
\newtheorem{lemma}[lemma]{Lemma}
\aliascntresetthe{lemma}

\newaliascnt{proposition}{theorem}

\aliascntresetthe{proposition}

\newaliascnt{claim}{theorem}

\aliascntresetthe{claim}

\theoremstyle{definition}
\newaliascnt{definition}{theorem}
\newtheorem{definition}[definition]{Definition}
\aliascntresetthe{definition}

\theoremstyle{remark}
\newaliascnt{remark}{theorem}

\aliascntresetthe{remark}

\newcommand{\cI}{\mathcal{I}}
\newcommand{\bX}{\mathbf{X}}
\newcommand{\bY}{\mathbf{Y}}
\newcommand{\btheta}{\boldsymbol{\theta}}
\newcommand{\bw}{\mathbf{w}}
\newcommand{\bv}{\mathbf{v}}

\newcommand{\WEFone}{\mathrm{WEF1}}
\newcommand{\WEFoneT}{\mathrm{WEF1T}}
\newcommand{\R}{\mathbb{R}}
\newcommand{\dd}{\,\mathrm{d}}

\DeclareMathOperator*{\argmax}{arg\,max}
\DeclareMathOperator*{\argmin}{arg\,min}

\crefname{theorem}{Theorem}{Theorems}
\Crefname{theorem}{Theorem}{Theorems}
\crefname{lemma}{Lemma}{Lemmas}
\Crefname{lemma}{Lemma}{Lemmas}
\crefname{proposition}{Proposition}{Propositions}
\Crefname{proposition}{Proposition}{Propositions}
\crefname{claim}{Claim}{Claims}
\Crefname{claim}{Claim}{Claims}
\crefname{definition}{Definition}{Definitions}
\Crefname{definition}{Definition}{Definitions}
\crefname{remark}{Remark}{Remarks}
\Crefname{remark}{Remark}{Remarks}
\crefname{algorithm}{Algorithm}{Algorithms}
\Crefname{algorithm}{Algorithm}{Algorithms}

\title{Almost Envy-Freeness for Additive Mixed Manna with Entitlements: Deterministic and Randomized Guarantees}
\author{%
  \begin{tabular}{@{}c@{\qquad\qquad}c@{}}
    Zehan Lin & Shengxin Liu \\
    University of Macau & Harbin Institute of Technology, Shenzhen \\
    \texttt{yc47490@um.edu.mo} & \texttt{sxliu@hit.edu.cn} \\
    \noalign{\vskip 1.0em}
    Biaoshuai Tao & Shengwei Zhou \\
    Shanghai Jiao Tong University & Nanyang Technological University \\
    \texttt{bstao@sjtu.edu.cn} & \texttt{s.arthur.zhou@gmail.com}
  \end{tabular}%
}
\date{}

\begin{document}
\maketitle

\begin{abstract}
We investigate the fair allocation of indivisible items among agents with asymmetric entitlements in mixed manna settings, where the items consist of both goods and chores.
For additive valuations, we establish that weighted envy-free up to one item ($\WEFone$) allocations always exist and can be computed in polynomial time.
We also study fair and efficient allocation and show that weighted envy-freeness up to one transfer ($\WEFoneT$) is compatible with fractional Pareto optimality (fPO) for every mixed-manna instance.
This relaxation from $\mathrm{WEF1}$ to $\mathrm{WEF1T}$ is tight, as demonstrated by our impossibility result.
We further show a best-of-both-worlds result via a finite lottery that guarantees weighted envy-freeness (WEF) in expectation, with every realized allocation satisfying $\mathrm{WEF1T}$ and achieving the tight characterization complemented by the existing impossibility result.
\end{abstract}

\section{Introduction}\label{sec:introduction}

Fair division, introduced by Steinhaus~\cite{steinhaus1948problem}, is a classic problem in economics and computer science concerning the allocation of a set $M$ of $m$ items among a set $N$ of $n$ agents with heterogeneous preferences~\cite{journals/ai/AmanatidisABFLMVW23}.
When item valuations are positive, the items are regarded as goods (e.g., valuable resources); when negative, they are considered chores (e.g., undesirable tasks).
If an instance contains both goods and chores, it is referred to as a \emph{mixed manna}~\cite{journals/corr/BogomolnaiaMSY17}.
The central goal of fair division is to find an allocation that guarantees fairness among all agents.
In this paper, we focus on the fair allocation of mixed manna, where each agent $i$ evaluates items via an additive valuation function $v_i$.

Several notions have been proposed to measure the fairness of allocations.
Among them, envy-freeness (EF)~\cite{foley1966resource} is one of the most well-studied, which requires every agent to weakly prefer her own bundle to every other agent's bundle.
Unfortunately, EF allocations may fail to exist when items are indivisible.
To address this, envy-freeness up to one item (EF1) was introduced as a standard relaxation of EF~\cite{conf/sigecom/LiptonMMS04}.
For goods, EF1 requires that any envy between two agents can be eliminated by removing an item from the envied agent’s bundle, whereas for chores, envy is eliminated by removing an item from the envious agent’s bundle.
In mixed manna settings, EF1 combines both principles and requires that envy can be eliminated by removing either a good from the envied agent's bundle or a chore from the envious agent's bundle~\cite{journals/aamas/AzizCIW22,conf/approx/BhaskarSV21}.

While much of the classical fair-division literature assumes that agents have equal entitlements, real-world settings often involve agents with different levels of responsibility or claims~\cite{journals/ipl/Suksompong25}.
The weighted model captures such asymmetry by assigning each agent $i\in N$ a positive weight $w_i$, with $\sum_{i\in N} w_i=1$, representing her entitlement or share of the total obligation; the classical unweighted setting is recovered when $w_i=1/n$ for every agent.
For indivisible goods, Chakraborty et al.~\cite{journals/teco/ChakrabortyISZ21} introduced \emph{weighted envy-freeness up to one item} (WEF1) and showed that such allocations can be computed in polynomial time for additive valuations.
Weighted picking-sequence methods and related relaxations of weighted envy-freeness were further studied in~\cite{journals/ai/ChakrabortySS21,journals/teco/ChakrabortySS24}.
For indivisible chores, polynomial-time WEF1 guarantees were subsequently established by Springer et al.~\cite{conf/aaai/SpringerHY24} and Wu et al.~\cite{journals/ai/WuZZ25}.
However, for mixed manna, the existence of WEF1 allocations remains open, with the exception of the two-agent case resolved by Garg and Sharma~\cite{journals/corr/abs-2410-12966}.
This naturally leads to the following question, which has also been raised in the literature~\cite{journals/ai/WuZZ25,journals/corr/abs-2410-12966,journals/ipl/Suksompong25}:
\begin{center}\begin{minipage}{0.97\linewidth}
\begin{shaded}
\noindent\textbf{Question 1.}
\textit{Are WEF1 allocations always guaranteed to exist for mixed manna, and if so, can they be computed in polynomial time?}
\end{shaded}
\end{minipage}
\end{center}

Fairness is only one of the desirable properties of an allocation.
Another important objective is efficiency, often captured by Pareto optimality (PO), meaning that no alternative allocation can make every agent weakly better off and at least one agent strictly better off.
Fractional Pareto optimality (fPO) applies the same comparison to fractional allocations and is therefore a stronger requirement.
Although weighted fairness and Pareto optimality have been studied jointly for goods and chores~\cite{journals/teco/ChakrabortyISZ21,journals/corr/EF1andPO}, their compatibility in the mixed-manna setting is less understood.
More recently, Mackenzie and Suzuki~\cite{journals/corr/abs-2607-17811} showed that, for a particular additive mixed-manna instance in the unweighted setting, no allocation is both EF1 and fPO, even though that same instance admits an allocation that is both EF1 and PO.
This negative result motivates the following question:
\begin{center}\begin{minipage}{0.97\linewidth}
\begin{shaded}
\noindent\textbf{Question 2.}
\textit{For mixed manna with arbitrary positive entitlements, can a natural relaxation of WEF1 always be achieved together with fPO?}
\end{shaded}
\end{minipage}
\end{center}

While previous work has mainly focused on deterministic allocations, another line of work studies randomized guarantees, particularly lotteries that combine ex-ante fairness with ex-post approximate fairness.
For indivisible goods, Aziz et al.~\cite{journals/ior/AzizFSV24} first established the coexistence of ex-ante EF and ex-post EF1 for equal entitlements.
In the weighted setting, Aziz et al.~\cite{conf/atal/0001GM23} and Hoefer et al.~\cite{journals/jair/HoeferSV24} established lotteries that are ex-ante weighted envy-free (WEF) and ex-post weighted envy-free up to one transfer (WEF1T) for indivisible goods.
Similar results have also been established for chores~\cite{journals/ai/WuZZ25}.
More recently, ex-ante EF and ex-post EF1 have been shown to coexist for unweighted mixed-manna instances~\cite{journals/corr/abs-2607-10232}.
These developments motivate the following question for weighted mixed-manna instances:
\begin{center}\begin{minipage}{0.97\linewidth}
\begin{shaded}
\noindent\textbf{Question 3.}
\textit{Does every mixed-manna instance admit a lottery that is ex-ante WEF and ex-post $\WEFoneT$?}
\end{shaded}
\end{minipage}
\end{center}

\subsection{Our Contributions}

We resolve all three questions above and, taken together, obtain a sharp picture of weighted envy-freeness for mixed manna.
Without imposing an additional objective, the standard one-item relaxation is always attainable: a WEF1 allocation can be computed in polynomial time.
By contrast, WEF1 cannot in general coexist with fPO, nor can it be guaranteed ex post by a lottery that is ex-ante WEF; in both settings, the one-transfer relaxation WEF1T is sufficient and is best possible within the family WEF$(x,y)$.
Thus, our results identify the precise relaxation threshold under both additional objectives.

Our first result resolves Question~1 for arbitrary positive entitlements and additive valuations.


\begin{center}\begin{minipage}{0.97\linewidth}
\begin{shaded}
\noindent\textbf{Result 1. (Theorem~\ref{thm:main})}
\textit{Given any mixed manna instance, there exists an algorithm that computes WEF1 allocations in polynomial time.}
\end{shaded}
\end{minipage}
\end{center}
Round-robin and picking-sequence algorithms are fundamental tools in fair division~\cite{journals/aamas/AzizCIW22,journals/ai/ChakrabortySS21,journals/teco/ChakrabortyISZ21,conf/sigecom/FeigeH23}.
For unweighted agents, round-robin guarantees EF1 for both goods and chores.
Aziz et al.~\cite{journals/aamas/AzizCIW22} extended this guarantee to mixed manna through double round-robin: chores are allocated by round-robin, followed by the remaining goods in reverse order.
In the weighted setting, weighted picking sequences similarly guarantee WEF1 for goods~\cite{journals/teco/ChakrabortyISZ21} and, with suitable modifications, for chores~\cite{journals/ai/WuZZ25,conf/aaai/SpringerHY24}.
In particular, the reversed weighted picking sequence of Wu et al.~\cite{journals/ai/WuZZ25} can be viewed as the weighted counterpart of the reversed round-robin step.
This connection naturally suggests combining forward and reversed weighted picking sequences for mixed manna.
However, Garg and Sharma~\cite{journals/corr/abs-2410-12966} showed that this direct combination guarantees only the weaker notion of WEF1T.

Our result shows that the double weighted picking-sequence approach can in fact achieve WEF1, provided that the instance is first appropriately preprocessed.
The preprocessing is based on the meta-good technique recently developed for mixed manna~\cite{journals/corr/abs-2607-10232,journals/corr/abs-2511-04891,journals/corr/abs-2607-10089}.
The main technical issue is that a meta-good may contain several original items, whereas WEF1 must be witnessed by the removal of a single original item.
We therefore construct the meta-goods carefully so that the guarantee produced by the weighted picking sequences can still be lifted, after unpacking, to a valid one-item comparison in the original instance.
This preprocessing resolves precisely the obstruction encountered by the direct approach and allows the forward and reversed weighted picking sequences to be combined without weakening WEF1 to WEF1T.
Together, these ingredients yield a WEF1 allocation for every mixed-manna instance.

Turning to efficiency, WEF1 and fPO cannot always be achieved simultaneously, even in the unweighted setting~\cite{journals/corr/abs-2607-17811}.
We show that WEF1T always suffices.


\begin{center}\begin{minipage}{0.97\linewidth}
\begin{shaded}
\noindent\textbf{Result 2. (Theorem~\ref{thm:wef1t-fpo})}
\textit{Given any mixed manna instance, there exists an allocation that satisfies both $\WEFoneT$ and fPO.}
\end{shaded}
\end{minipage}
\end{center}
We first reduce the instance to objective goods and objective chores~\cite{conf/ijcai/0001X025}., and augment it with private items and zero-valued dummy items.
This augmentation provides the boundary structure required by the Knaster--Kuratowski--Mazurkiewicz (KKM) argument~\cite{knaster1929beweis,conf/wine/BarmanHSS25,journals/geb/IgarashiM26,journals/corr/abs-2509-18673,journals/corr/EF1andPO}.
Each point of the simplex defines a perturbed weighted-welfare maximization problem and corresponding price comparisons.
KKM selects a weight vector for which every agent has a valid comparison, but the comparisons for different agents may be witnessed by different welfare maximizers.
The main technical step is to compare these different maximizers and use the tie-breaking potential to extract one allocation satisfying all comparisons simultaneously, which yields $\WEFoneT$.
Letting the perturbation vanish then gives fPO, and the reduction and projection preserve $\WEFoneT$.
Thus, KKM selects the weights, while the tie-breaking step turns the agent-wise certificates into a single fair and efficient allocation.

For the third question, we study whether the fairness guarantee can hold simultaneously in expectation and in every realization.

\begin{center}\begin{minipage}{0.97\linewidth}
\begin{shaded}
\noindent\textbf{Result 3. (Theorem~\ref{thm:bobw})}
\textit{Given any mixed manna instance, there exists a lottery that is ex-ante WEF and ex-post $\WEFoneT$.}
\end{shaded}
\end{minipage}
\end{center}

To obtain Result 3, we first reduce the mixed-manna instance to a collection of meta-goods and residual chores that are strongly negative to every agent.
The main technical challenge is to couple the two parts: independently combining a residual-chore lottery with a meta-good lottery preserves their marginal guarantees, but need not yield WEF1T in every realization.
For the residual chores, we start from the DSE-based lottery of Wu et al.~\cite{journals/ai/WuZZ25}.
We then use a maximum-entropy argument~\cite{conf/stoc/SinghV14} to select a structured integral decomposition and derive probability bounds on which agents receive an extra chore.
These bounds allow us to verify a fractional Hall condition~\cite{hall1935representatives,schrijver2003combinatorial} and thereby construct, conditional on each realized chore allocation, a compatible randomized allocation of the meta-goods while preserving the entitlement-proportional marginals required for ex-ante WEF.
At the same time, the conditional allocation guarantees the normalized chore-count comparisons needed for the one-transfer argument.
Together with the strong negativity of the residual chores, this yields WEF1T in every realization.
Finally, the defining properties of the meta-goods convert every meta-level witness into a transfer involving a single original item.
The resulting finite lottery is therefore WEF in expectation and ex-post WEF1T.

\subsection{Related Work}
Besides EF1, several other fairness notions have been extensively studied in the literature.
Envy-freeness up to any item (EFX) is a stronger relaxation of envy-freeness~\cite{journals/teco/CaragiannisKMPS19}, while other prominent notions include maximin share (MMS)~\cite{conf/bqgt/Budish10} and proportionality up to one item (PROP1)~\cite{conf/sigecom/ConitzerF017}.
Given the extensive literature on fair division, we focus on work most closely related to fair allocation for weighted settings and mixed manna.
For broader background, we refer the reader to the surveys~\cite{journals/sigecom/AzizLMW22,journals/ai/AmanatidisABFLMVW23}.

\paragraph{Weighted Settings.}
Beyond the existence and efficient computation of WEF1 allocations, a central line of research studies whether weighted fairness can be achieved together with Pareto optimality (PO).
For indivisible goods, Chakraborty et al.~\cite{journals/teco/ChakrabortyISZ21} showed that a WEF1 and PO allocation always exists for additive valuations and can be computed in pseudo-polynomial time.
For indivisible chores, Wu et al.~\cite{journals/ai/WuZZ25} established polynomial-time algorithms for computing WEF1 and PO allocations in several special cases, including bi-valued instances and instances with two agents.
Garg et al.~\cite{conf/ijcai/GargMQ24} further obtained polynomial-time WEF1 and PO allocations for instances with three types of agents or two types of chores.
Most recently, Mahara~\cite{journals/corr/EF1andPO} resolved the general existence question for additive chores, showing that WEF1 and fractional Pareto optimality can always be achieved simultaneously.
Related weighted fairness notions have also received extensive study, e.g., \emph{weighted} PROP1/PROPX~\cite{journals/orl/AzizMS20,journals/ai/AzizLMWZ24}, \emph{weighted} MMS~\cite{journals/jair/FarhadiGHLPSSY19,conf/ijcai/0001C019,conf/ijcai/00020L24,journals/corr/abs-2510-06581,journals/corr/abs-2510-10698}, and \emph{AnyPrice share} (APS)~\cite{conf/sigecom/BabaioffEF21,conf/sigecom/FeigeH23}.
A recent survey of weighted fair division is given by Suksompong~\cite{journals/ipl/Suksompong25}.

\paragraph{Mixed Manna.}
It is worth mentioning that standard algorithmic techniques developed for purely goods or chores often break down when applied to mixed settings~\cite{journals/aamas/AzizCIW22, conf/approx/BhaskarSV21, journals/mor/ChaudhuryGMM23}, which requires novel and independent analytical frameworks~\cite{journals/corr/abs-2501-06799}.
For mixed manna under additive valuations, Bhaskar et al.~\cite{conf/approx/BhaskarSV21} and Aziz et al.~\cite{journals/aamas/AzizCIW22} established the existence and polynomial-time computation of EF1 allocations.
To achieve efficiency, Barman et al.~\cite{conf/wine/BarmanHSS25} introduced envy-freeness up to $k$ reallocations (EFR-$k$), which requires that envy can be eliminated by reassigning a subset of at most $k$ items, and showed the existence of EFR-$(n-1)$ and PO allocations.
Subsequently, Barman and Verma~\cite{journals/corr/abs-2509-18673} improved this result to introspectively envy-free up to one item (IEF1), where any agent can eliminate envy by adding or removing a single item from her own bundle, and established the general existence of allocations that are both PO and IEF1.
There are also works that study other fairness notions for the mixed setting, e.g., for EFX~\cite{conf/ki/AleksandrovW20, conf/ijcai/HosseiniMW23, conf/atal/HosseiniSVX23}, MMS~\cite{feige2022maximin, conf/wine/CousinsVZ23, conf/sigecom/KulkarniMT21} and PROP1~\cite{journals/orl/AzizMS20}.
See~\cite{journals/jair/LiuLSW24} for more detailed reviews for the mixed setting.

\subsection{Organization}
We begin in Section~\ref{sec:model} by formalizing the model and the WEF1 notion.
Section~\ref{sec:computation} presents the computation of WEF1 allocations: it introduces meta-goods, constructs them in Section~\ref{subsec:construction}, develops the weighted picking sequences in Section~\ref{sec:picking}, and uses them to allocate the reduced instance in Section~\ref{sec:remaining}.
Section~\ref{sec:fpo} turns to efficiency: it first gives an impossibility result for $\WEFone$ and fractional Pareto optimality, and then establishes the existence of a $\WEFoneT$ allocation that is fractionally Pareto optimal.
Section~\ref{sec:bobw} establishes a lottery that is ex-ante WEF and ex-post $\WEFoneT$.
Finally, Section~\ref{sec:conclusion} concludes with open problems; deferred proofs are given in the appendix.

\section{Preliminary}\label{sec:model}
We study the problem of fairly allocating a set $M$ of $m$ indivisible items, which consists of both goods and chores, among a set $N$ of $n$ agents.
In this setting, each agent $i$ has a positive entitlement $w_i>0$.
Without loss of generality, we assume the weights are normalized such that $\sum_{i \in N} w_i = 1$.
When $w_i = 1/n$ for all $i \in N$, the instance is referred to as unweighted.
A subset of items $B \subseteq M$ is referred to as a bundle.
Each agent $i \in N$ has an additive valuation function $v_i: 2^M \to \mathbb{R}$, which assigns a value to every bundle $B \subseteq M$.
For ease of notation, we write $v_i(e)$ for $v_i(\{e\})$.
Thus, $v_i(B)=\sum_{e\in B}v_i(e)$ for every $B\subseteq M$, and $v_i(\varnothing)=0$.
Throughout this paper, we use $\mathbf{w} = (w_1, \ldots, w_n)$ and $\mathbf{v} = (v_1, \ldots, v_n)$ to denote the entitlement vector and valuation profile of the agents, respectively.
An allocation $\bX = (X_1, \dots, X_n)$ is an ordered partition of the item set $M$ into $n$ disjoint bundles, where agent $i$ receives bundle $X_i$.
Given an instance $\cI = (N, M, \bw, \bv)$, our goal is to find an allocation $\bX$ that is fair to all agents.
For any nonnegative integer $k$, let $[k]=\{1,2,\ldots,k\}$, where $[0]=\varnothing$.

For any agent $i \in N$, items can be classified as \emph{goods} or \emph{chores} according to their values under $v_i$.
Specifically, an item $e$ is a good to agent $i$ if $v_i(e) \geq 0$, a chore if $v_i(e) < 0$.
We note that these classifications may differ across agents.
In other words, an item can be a good to some agent but a chore to another agent.
Furthermore, an item $e$ is called an \emph{objective good} (resp., \emph{objective chore}) if $v_i(e) \ge 0$ (resp., $v_i(e) < 0$) for all agents $i \in N$.
Finally, an item $e$ is a \emph{subjective good} if $v_i(e) \ge 0$ for at least one agent $i \in N$.

In the following, we introduce several fairness notions.

\begin{definition}[WEF]\label{def:wef}
An allocation $\bX$ satisfies weighted envy-freeness (WEF) if for any two agents $i,j\in N$, we have $v_i(X_i)/w_i \ge v_i(X_j)/w_j$.
\end{definition}

While WEF provides a strong fairness guarantee, such allocations may not always exist, which motivates us to consider two relaxed fairness notions derived from WEF.
\begin{definition}[WEF1]\label{def:wef1}
An allocation $\bX$ satisfies weighted envy-freeness up to one item (WEF1) if, for any two agents $i,j\in N$, at least one of the following holds:
\begin{enumerate}[label=(\roman*),ref=(\roman*)]
    \item agent $i$ satisfies WEF towards $j$: $v_i(X_i) / w_i \ge v_i(X_j) / w_j$;
    \item there exists an item $e\in X_j$ with $v_i(e)>0$ such that $v_i(X_i) / w_i \ge v_i(X_j\setminus\{e\}) / w_j$;
    \item there exists an item $e\in X_i$ with $v_i(e)<0$ such that $v_i(X_i\setminus\{e\}) / w_i \ge v_i(X_j) / w_j$.
\end{enumerate}
\end{definition}

In particular, we say that agent $i$ satisfies $\mathrm{WEF1}$ towards agent $j$ if any of the above conditions holds for two agents $i, j \in N$.
\begin{definition}[WEF1T]\label{def:wef1t}
An allocation $\bX$ satisfies \emph{weighted envy-freeness up to one transfer} (WEF1T), if, for any two agents $i, j \in N$, at least one of the following holds:
\begin{enumerate}[label=(\roman*),ref=(\roman*)]
    \item agent $i$ is WEF towards $j$: $v_i(X_i) / w_i \ge v_i(X_j) / w_j$;
    \item there exists an item $e\in X_j$ with $v_i(e)>0$ such that $v_i(X_i \cup \{e\}) / w_i \ge v_i(X_j\setminus\{e\}) / w_j$;
    \item there exists an item $e\in X_i$ with $v_i(e)<0$ such that $v_i(X_i\setminus\{e\}) / w_i \ge v_i(X_j \cup \{e\}) / w_j$.
\end{enumerate}
\end{definition}

Next, we formally define the efficiency notion considered in this paper.
\begin{definition}[PO and fPO]
An allocation $\bX'$ \emph{Pareto dominates} another allocation $\bX$ if $v_i(X_i') \geq v_i(X_i)$ for all $i \in N$, with strict inequality for at least one agent.
An integral allocation $\bX$ is \emph{Pareto optimal} (PO) if it is not Pareto dominated by any other integral allocation.
Moreover, $\bX$ is \emph{fractional Pareto optimal} (fPO) if it is not Pareto dominated by any fractional allocation\footnote{In a fractional allocation, items can be infinitely divided.
Formally, it is represented by variables $x_{ij} \in [0, 1]$ denoting the fraction of item $j$ allocated to agent $i$, such that $\sum_{i \in N} x_{ij} = 1$ for every item $j$.
An agent's valuation extends linearly, i.e., $v_i(\bX_i) = \sum\limits_{j \in M} x_{ij} \cdot v_i(j)$.}.
\end{definition}


\section{The Computation of WEF1 Allocations}\label{sec:computation}\label{sec:metagoods}\label{sec:construction}
In this section, we introduce our algorithm that computes WEF1 allocation for the mixed setting.
First, we provide an overview of our algorithm.


\paragraph{Algorithm Overview.}
Our algorithmic framework operates in two distinct phases.
The first phase is a bundling preprocessing step based on the meta-good technique, which has recently emerged in a line of research on mixed manna~\cite{journals/corr/abs-2511-04891, journals/corr/abs-2607-10089, journals/corr/abs-2607-10232}.
The high-level idea is to group several items into a bundle called a \emph{meta-good} that will be allocated completely.
By design, this construction guarantees desirable structural properties while leaving behind a subset of objective chores, referred to as residual chores (see Algorithm~\ref{alg:metagood-construction}).
In the second phase, we establish weighted picking sequences adapted for both goods and chores.
By carefully combining these picking sequences with the underlying properties of the meta-goods, we jointly allocate meta-goods and residual chores to achieve WEF1 allocations (Algorithm~\ref{alg:reduced-allocation}).
\medskip

To begin with, we formally define the concept of a meta-good.
\begin{definition}[Meta-good]\label{def:metagood}
Fix a nonempty set of agents $N'\subseteq N$.
A nonempty set $G\subseteq M$ is a \emph{meta-good over $N'$} if
\begin{enumerate}[label=(\roman*),ref=(\roman*)]
    \item $\max_{i\in N'} \{v_i(G)\}\ge0 $;
    \item for every $i\in N'$ with $v_i(G) \ge 0$, there exists an original item
$e\in G$ such that
    $$
    v_i(e) \ge 0
    \qquad\text{and}\qquad
    v_i(G\setminus\{e\})\le0.
    $$
\end{enumerate}
\end{definition}

The intuition behind this definition is twofold. 
Condition (i) ensures that $G$ can be assigned to an agent who values it non-negatively.
Moreover, condition (ii) is essential to guarantee the $\WEFone$ property. 
To see this, consider any two agents $i, j \in N'$, and suppose a meta-level allocation (where each meta-good is regarded as an indivisible item) assigns meta-good $G$ to agent $j$ such that
$$
\frac{v_i(X_i)}{w_i}
\ge
\frac{v_i(X_j\setminus G)}{w_j}
\qquad\text{and}\qquad v_i(G)>0.
$$
By condition (ii), we can select a single item $e\in G$ such that
$$
v_i(X_j\setminus\{e\})
=v_i(X_j\setminus G)+v_i(G\setminus\{e\})
\le v_i(X_j\setminus G).
$$
Thus, removing the single good $e$ is at least as helpful as deleting the whole meta-good.
Notice that every singleton $\{e\}$ with $\max_{i\in N'} \{v_i(e)\}\ge0$ is a meta-good over $N'$.

\paragraph{Remark.}
While Lu et al.~\cite{journals/corr/abs-2607-10089} require a similar upper-bound condition for meta-goods, their construction does not necessarily satisfy condition (ii) above.
To illustrate this, consider a simple instance with two agents and three items, whose valuations are given in Table~\ref{tab:meta_good_counterexample}.
Applying Algorithm~3 from~\cite{journals/corr/abs-2607-10089} to this instance merges $\{e_1, e_2\}$ into a meta-good, which fails to satisfy condition (ii) in Definition~\ref{def:metagood}.

\begin{table}[htbp]
    \centering
    \begin{tabular}{c|c|c|c}
            & $e_1$ & $e_2$ & $e_3$ \\
    \hline
    Agent 1 & 0.4   & 0.4   & -0.9  \\
    Agent 2 & -1    & -1    & -0.9
    \end{tabular}
    \caption{Counterexample for the meta-good construction in~\cite{journals/corr/abs-2607-10089}.}
    \label{tab:meta_good_counterexample}
\end{table}

Thus, in the following, we present a modification of the construction that satisfies our strengthened definition of meta-goods.
Similar to Algorithm 3 in~\cite{journals/corr/abs-2607-10089}, we start from singleton positive items and first construct a disjoint family of meta-goods by repeatedly absorbing strictly negative items.
This leaves residual chores with a strong negative value.



\subsection{The Construction of Meta-Goods}\label{subsec:construction}
In this subsection, we show how to construct meta-goods that satisfy the desirable properties in Definition~\ref{def:metagood}.
To facilitate our construction, we partition the item set $M$ into two disjoint subsets as below:
$$
M^+ =\{e\in M:\max_{i\in N} \{v_i(e)\} \ge0\},
\qquad
M^- =\{e\in M:\max_{i\in N} \{v_i(e)\} <0\}.
$$
By definition, every item in $M^-$ is an objective chore for each agent, whereas every singleton $\{e\}$ with $e \in M^+$ trivially forms a meta-good over $N$.

Throughout the construction, we use $\mathcal{G}$ to denote the collection of meta-goods and $\mathcal{D}$ to denote the set of objective chores.
For each agent $i \in N$, let $G^i_1, G^i_2, \ldots, G^i_{\ell_i}$ denote the meta-goods with $v_i(G^i_k) \ge 0$ for $k \in [\ell_i]$, indexed in non-increasing order of value to $i$ by $v_i(G^i_1)\ge v_i(G^i_2)\ge\cdots\ge v_i(G^i_{\ell_i}) \ge 0$.
For any item $e\in \mathcal D$, we further define
\begin{equation*}
 k_{i,e}:=
 \min\left\{k\in[\ell_i]:
 v_i(e)+\sum_{t=1}^k v_i(G^i_t)\ge0\right\},
\end{equation*}
with $k_{i,e}=+\infty$ if the displayed set is empty.
Intuitively, $k_{i,e}$ is the shortest prefix of $i$'s nonnegative meta-goods that compensates for the disutility of $e$, which allows us to merge several meta-goods with an objective chore.

We are now ready to present the detailed construction procedure in Algorithm~\ref{alg:metagood-construction}.
The algorithm starts with each item in $M^+$ as a singleton meta-good and places all objective chores in $\mathcal D$.
In each iteration, it updates $k_{i,e}$, selects a feasible pair $(i, e)$ that minimizes this value, and merges $e$ with the corresponding minimal prefix of meta-goods.
This process terminates when no remaining chore in $\mathcal{D}$ can be absorbed.

\begin{algorithm}[htbp]
\caption{Meta-Good Construction}\label{alg:metagood-construction}
\KwIn{Agents $N$, items $M$, and additive valuations $(v_i)_{i\in N}$.}
Set $\mathcal G\gets\{\{g\}:g\in M^+\}$ and $\mathcal D\gets M^-$\;
\While{\textnormal{true}}{
    For every $i\in N$, compute $k_{i,e}$, for $e\in \mathcal D$ as definition\;
    \If{$k_{i,e}=+\infty$ for every $i \in N$ and $e \in \mathcal D$}{
        \textbf{break}\;
    }
    Choose a feasible pair $(i,e)$ minimizing $k_{i,e}$\;
    Set $\mathcal S\gets\{G^i_1,\ldots,G^i_{k_{i,e}}\}$ and
    $G'\gets\{e\}\cup\bigcup_{G\in\mathcal S}G$\;
    Update $\mathcal G\gets(\mathcal G\setminus\mathcal S)\cup\{G'\}$ and
    $\mathcal D\gets \mathcal D\setminus\{e\}$\;
}
\KwOut{$(\mathcal G, \mathcal D)$.}
\end{algorithm}

By construction, Algorithm~\ref{alg:metagood-construction} provides several structural guarantees upon termination.
For ease of notation, for any subfamily $\mathcal G'\subseteq\mathcal G$, we write $v_i(\mathcal G')=\sum_{G\in\mathcal G'}v_i(G)$.

\begin{lemma}\label{lem:construction}
Upon termination of Algorithm~\ref{alg:metagood-construction}, the following properties hold:
\begin{enumerate}[label=(\roman*), ref=(\roman*)]
    \item Every element in $\mathcal{G}$ is a meta-good over $N$.
    \item For every agent $i \in N$, every chore $e \in \mathcal{D}$, we have $v_i(e) + v_i(\mathcal G') < 0$ for any $\mathcal G' \subseteq \mathcal G$.
\end{enumerate}
\end{lemma}
\begin{proof}
We first show that each iteration preserves the fact that every element of $\mathcal G$ is a meta-good over $N$.
This is true initially because every singleton $\{e\}$ with $e\in M^+$ is a meta-good.
Consider an iteration that selects $(i,e)$, and let $\mathcal S=\{G^i_1,\ldots,G^i_k\}$, and $G'=\{e\}\cup\bigcup_{G\in\mathcal S}G$ as illustrated in the algorithm.
In the following, we show that $G'$ is a meta-goods over $N$.

By the design of the algorithm, we have $v_i(G')\ge0$, which guarantees condition (i) in Definition~\ref{def:metagood}.
We next show that $G'$ also satisfies condition (ii) in Definition~\ref{def:metagood}.
Fix any agent $j\in N$ with $v_j(G')\ge0$.
By the choice of $(i,e)$ that minimizes $k_{i,e}$, we must have $v_j(G'\setminus G)<0$ for every $G\in\mathcal S$. Otherwise, suppose $v_j(G'\setminus G) \ge 0$ for some $G\in\mathcal S$. The remaining non-negative meta-goods in $\mathcal S\setminus\{G\}$ (at most $k_{j,e}-1$ of them), together with $e$, would still yield a non-negative total value for agent $j$. This implies $j$ could find a feasible prefix of length at most $k_{i,e}-1$, yielding $k_{j,e} \le k_{i,e} - 1 < k_{i,e}$, which contradicts the minimality of $k_{i,e}$.

It follows that $v_j(G)=v_j(G')-v_j(G'\setminus G)>0$ for every $G\in\mathcal S$.
Since $G$ is a meta-good, there is an item $e'\in G$ such that $v_j(e')\ge0$ and $v_j(G\setminus\{e'\})\le0$.
Therefore,
$$
v_j(G'\setminus\{e'\})
=v_j(G'\setminus G)+v_j(G\setminus\{e'\})<0,
$$
which ensures condition (ii).
Hence $G'$ is a meta-good and every iteration preserves the required property.
At termination, it follows that $k_{i,e} = +\infty$ for every $i \in N$ and $e \in \mathcal{D}$.
Thus, for any $\mathcal G' \subseteq \mathcal G$, we have
$$
v_i(e) + v_i(\mathcal G') \leq v_i(e) + \sum_{G \in \mathcal G:\, v_i(G) \ge 0} v_i(G) < 0.
$$
This completes the proof.
\end{proof}

Therefore, it remains to show that Algorithm~\ref{alg:metagood-construction} is guaranteed to terminate efficiently.
\begin{lemma}\label{lem:construction-runtime}
\Cref{alg:metagood-construction} terminates in $O(nm^2 \log m)$.
\end{lemma}

\begin{proof}
Since each iteration removes one item from $\mathcal{D}$, Algorithm~\ref{alg:metagood-construction} terminates after at most $m$ iterations.
Throughout the execution, the family $\mathcal{G}$ contains at most $m$ meta-goods.
For a fixed agent $i \in N$, sorting the current meta-goods in non-increasing order w.r.t. some $v_i$ and constructing the prefix sums requires $O(m \log m)$ time.
Since the prefix sums are monotonically non-decreasing, each index $k_{i,e}$ can be located via binary search in $O(\log m)$ time.
Consequently, computing $k_{i,e}$ for all $i \in N$ and $e \in \mathcal{D}$ takes $O(nm \log m)$ time.
In addition, identifying a feasible pair $(i,e)$ that minimizes $k_{i,e}$ takes $O(nm)$ time, and executing the corresponding merge operation requires $O(m)$ time.
Therefore, each iteration runs in $O(nm \log m)$ time, which yields an overall runtime of $O(nm^2 \log m)$ for Algorithm~\ref{alg:metagood-construction}.
\end{proof}

\subsection{Forward and Reversed Weighted Picking Sequence}\label{sec:picking}
Given the meta-goods $\mathcal G$ and residual chores $\mathcal D$ computed by Algorithm~\ref{alg:metagood-construction}, it remains to allocate these objects so that the resulting allocation is $\WEFone$.
Intuitively, the meta-goods in $\mathcal G$ play the role of goods (with respect to some agents), whereas the residual chores can be combined with meta-goods to form objects that are strictly negative for every agent.

To achieve this, we apply two weighted picking algorithms: the Forward Weighted Picking Sequence (FWPS) for goods~\cite{journals/teco/ChakrabortyISZ21} and the Reversed Weighted Picking Sequence (RWPS) for chores~\cite{journals/ai/WuZZ25}.
Specifically, the forward sequence allocates the meta-goods in $\mathcal G$, whereas the reversed sequence processes the grouped objects formed by combining each residual chore with the meta-goods valued non-negatively by the selected agent.

We first present the allocation procedure for meta-goods (Algorithm~\ref{alg:forward-picking}).
Unlike standard goods allocation, meta-goods may be valued negatively by some agents.
Intuitively, such meta-goods should not be assigned to agents who view them as chores.
To address this, we demonstrate that all meta-goods can be fully allocated while guaranteeing that every agent receives a bundle of non-negative utility (Lemma~\ref{lem:forward-metagoods}).
In the forward procedure, let $s_i$ denote the normalized pick count of agent $i$, initialized to $s_i = 0$ for all $i \in N'$.
At each step, we select an agent with the minimum $s_i$ among those who value at least one remaining meta-good positively.
The selected agent receives her most preferred remaining meta-good, and $s_i$ increases by $1/w_i$.
Finally, those meta-goods that yield non-positive utility to all agents are subsequently assigned to agents who value them at zero.

\begin{algorithm}[htbp]
\caption{Forward Weighted Picking Sequence}\label{alg:forward-picking}
\KwIn{Agents $N'\subseteq N$, entitlements $(w_i)_{i\in N'}$, and a meta-goods set $\mathcal G$ over $N'$.}
Initialize $Y_i\gets\varnothing$ and $s_i\gets0$ for every $i\in N'$; \\
Set $\mathcal G^0\gets\{G\in\mathcal G:v_i(G)\le0\text{ for every }i\in N'\}$ and $\mathcal R\gets\mathcal G\setminus\mathcal G^0$\;
\While{$\mathcal R\neq\varnothing$}{
    Let $N^0=\{j\in N':\text{some }G\in\mathcal R\text{ satisfies }v_j(G)>0\}$\;
    Choose $i\in\argmin_{j\in N^0} \{s_j\}$\;
    Choose $G\in\argmax_{G'\in\mathcal R}\{v_i(G')\}$\;
    Update $Y_i\gets Y_i\cup\{G\}$, $s_i\gets s_i+1/w_i$, and $\mathcal R\gets\mathcal R\setminus\{G\}$\;
}
\For{$G\in\mathcal G^0$}{
    Choose $i\in\argmax_{j\in N'} \{v_j(G)\}$ and update $Y_i\gets Y_i\cup\{G\}$\;
}
\KwOut{An allocation $\bY$.}
\end{algorithm}

The formal properties of Algorithm~\ref{alg:forward-picking} are summarized in the following lemma, with its proof deferred to Appendix~\ref{appendix: 3.2}.

\begin{lemma}\label{lem:forward-metagoods}
Given $N' \subseteq N$ and meta goods set $\mathcal{G}$ over $N'$, Algorithm~\ref{alg:forward-picking} computes a $\WEFone$ allocation among $N'$ such that $v_i(Y_i) \geq 0$ for every $i \in N'$.
\end{lemma}


We next present the RWPS algorithm that allocates a family $\mathcal{Q}$ of indivisible objects to a group of agents $N'$.
In our framework, each object may be a singleton objective chore or a bundle of base items $Q\in \mathcal Q$ such that $v_i(Q) < 0$ for all $i\in N'$.
In Algorithm~\ref{alg:reduced-allocation}, each object processed by RWPS is constructed by combining one residual chore with a set of meta-goods.
Unlike goods allocation, RWPS first computes a picking sequence $\sigma(t)$ and then executes the picking process in reverse order.
\begin{algorithm}[htbp]
\caption{Reversed Weighted Picking Sequence}\label{alg:reverse-picking}
\KwIn{Agents $N$, entitlements $(w_i)_{i\in N'}$, and a nonempty family $\mathcal Q$ of indivisible objects satisfying $v_i(Q)<0$ for every $i\in N'$ and $Q\in\mathcal Q$.}
Initialize $s_i\gets0$ and $Y_i\gets\varnothing$ for every $i\in N'$, and set $\mathcal R\gets\mathcal Q$\;
\For{$t=1,2,\ldots,|\mathcal Q|$}{
    Choose $\sigma(t)\in\arg\min_{j\in N'} \{s_j\}$ and let $i\gets\sigma(t)$\;
    Update $s_{i}\gets s_{i}+1/w_{i}$\;
}
\For{$t=|\mathcal Q|,|\mathcal Q|-1,\ldots,1$}{
    Let $i\gets\sigma(t)$\;
    Choose $Q^*\in\arg\max_{Q\in\mathcal R}\{v_i(Q)\}$\;
    Update $Y_i\gets Y_i\cup\{Q^*\}$ and $\mathcal R\gets\mathcal R\setminus\{Q^*\}$\;
}
\KwOut{An allocation $\bY$ of $\mathcal Q$.}
\end{algorithm}

The following lemma is obtained directly from Theorem~3.1 of~\cite{journals/ai/WuZZ25} by negating the valuation functions.

\begin{lemma}[{\cite{journals/ai/WuZZ25}}]\label{lem:reverse-picking}
For any family $\mathcal{Q}$ satisfying the input conditions of Algorithm~\ref{alg:reverse-picking}, the algorithm outputs an allocation $\mathbf{Y}$ such that for every pair of agents $i, j \in N'$, either $v_i(Y_i)/w_i \geq v_i(Y_j) / w_j$, or there exists some $Q \in Y_i$ such that $v_i(Y_i \setminus \{Q\}) / w_i \ge v_i(Y_j) / w_j$.
\end{lemma}

\subsection{Allocating the Reduced Instance}\label{sec:remaining}
We are now ready to complete the computation of WEF1 allocations.
Let $(\mathcal{G},\mathcal{D})$ be the output of Algorithm~\ref{alg:metagood-construction}.
The allocation procedure consists of two steps (see Algorithm~\ref{alg:reduced-allocation} for details).

\paragraph{Step 1: Joint Allocation of Meta-Goods and Chores.}
We first jointly allocate the meta-goods and residual chores to all agents in $N$ using the reversed picking sequence.
Specifically, the algorithm constructs the picking sequence and subsequently processes it in reverse order.
At round $t$, agent $\sigma(t)$ receives her most preferred remaining residual chore, together with all currently unassigned meta-goods for which she has non-negative valuation.
\paragraph{Step 2: Allocation of Remaining Meta-Goods.}
Following the joint allocation phase, any remaining meta-goods in $\mathcal{R}$ are distributed via Algorithm~\ref{alg:forward-picking} among agents who are absent from the sequence.
For simplicity, $\mathrm{FWPS}(N',\mathcal{G})$ denotes the execution of Algorithm~\ref{alg:forward-picking} on an agent subset $N' \subseteq N$ and a family of meta-goods $\mathcal{G}$.
Finally, the resulting allocation of original items can be obtained by unbundling the allocated meta-goods.
\medskip

\begin{algorithm}[!htb]
\caption{Allocating the Reduced Instance}\label{alg:reduced-allocation}
\KwIn{Agents $N$, entitlements $(w_i)_{i\in N}$, $(\mathcal G, \mathcal D)$ output by Algorithm~\ref{alg:metagood-construction}.}
Initialize $X_i\gets\varnothing$ and $s_i\gets0$ for every $i\in N$, and set $\mathcal D^0\gets\mathcal D$\;
\tcp{Step 1: Joint Allocation of Meta-Goods and Chores}
\For{$t=1,2,\ldots,|\mathcal D^0|$}{
    Choose $\sigma(t)\in\argmin_{i\in N} \{s_i\}$\;
    Update $s_{\sigma(t)}\gets s_{\sigma(t)}+1/w_{\sigma(t)}$\;
}
\For{$t=|\mathcal D^0|,|\mathcal D^0|-1,\ldots,1$}{
    Let $i\gets\sigma(t)$\;
    Choose $e_t\in\argmax_{e\in\mathcal D} \{v_i(e)\}$\;
    Set
    $\mathcal G_t\gets\{G\in\mathcal G:v_i(G)\ge0\}$\;
    Set
    $Q_t\gets\{e_t\}\cup\bigcup_{G\in\mathcal G_t}G$ and
    $X_i\gets X_i\cup Q_t$\;
    Update
    $\mathcal D\gets\mathcal D\setminus\{e_t\}$ and
    $\mathcal G\gets\mathcal G\setminus\mathcal G_t$\;
}

\tcp{Step 2: Allocation of Remaining Meta-Goods}
Let $\mathcal{R} \gets \mathcal{G}$;\\
Set $N'\gets\{\sigma(t):t=1,\ldots,|\mathcal D^0|\}$\;
\If{$\mathcal R\neq\varnothing$}{
    $\bY\gets\mathrm{FWPS}(N\setminus N',\mathcal G)$\;
    $X_i\gets X_i\cup Y_i$ for every $i\in N\setminus N'$\;
}
Update $\bX$ by unbundling all meta-goods to the original items\;
\KwOut{An WEF1 allocation $\bX$.}
\end{algorithm}

Next, we establish the fairness guarantees provided by Algorithm~\ref{alg:reduced-allocation}.
To facilitate the analysis, we first introduce the following notation.

Let $\mathcal{G}^0$ and $\mathcal{D}^0$ denote the sets of input meta-goods and residual chores, respectively.
For each $t \in \{1, \ldots, |\mathcal{D}^0|\}$, let $e_t$, $\mathcal{G}_t$, and $Q_t$ denote the objects constructed during iteration $t$ of the reverse loop.
We let $\mathcal{Q} = \{Q_t : t = 1, \ldots, |\mathcal{D}^0|\}$ denote the family of bundled objects created throughout this reverse loop.
Finally, at the end of Step 1 (i.e., after the reverse loop), let $\mathcal{R}$ denote the set of remaining meta-goods, and for each agent $i \in N$, let $X'_i$ denote the bundle assigned to agent $i$.

\begin{lemma}\label{lem:reverse-loop-structure}
At the end of Step 1, every pair of agents $i, j \in N$ satisfies the $\WEFone$ condition in the original instance.
\end{lemma}

\begin{proof}
For every $t\in\{1,\ldots,|\mathcal D^0|\}$, we have $e_t\in\mathcal D^0$ and $\mathcal G_t\subseteq\mathcal G^0$.
By Property~(ii) of Lemma~\ref{lem:construction}, we have $v_i(Q_t)<0$ for every $i\in N$.
Thus, $Q_t$ is strictly negative for every agent and is regarded as one indivisible object by RWPS.

Now suppose $1\le u<t\le|\mathcal D^0|$ and let $i=\sigma(t)$.
When iteration $t$ is executed, chore $e_u$ remains unselected because iteration $u$ is processed later in the reverse order.
By the greedy choice in Algorithm~\ref{alg:reduced-allocation}, it follows that $v_i(e_t)\ge v_i(e_u)$.
Note that by the design of the algorithm, every meta-good in $\mathcal{G}_t$ yields a non-negative value to agent $i$.
Furthermore, every meta-good in $\mathcal G_u$ is available during iteration $t$ but excluded from $\mathcal G_t$, which implies that all meta-goods in $\mathcal G_u$ yield strictly negative values to agent $i$.
Hence, we have
\begin{align*}
  v_i(Q_t)
  &=v_i(e_t)+\sum_{G\in\mathcal G_t}v_i(G) \ge v_i(e_t)
   \ge v_i(e_u)\\
  &\ge v_i(e_u)+\sum_{G\in\mathcal G_u}v_i(G)
   =v_i(Q_u),
\end{align*}
which implies that when iteration $t$ is processed, object $Q_t$ is the most preferred among all available objects for agent $i = \sigma(t)$.
Hence the reverse loop is precisely a valid execution of \cref{alg:reverse-picking} on $\mathcal Q$ with the order $\sigma(1),\ldots,\sigma(|\mathcal D^0|)$.

Next, we prove that $\bX'$ satisfies WEF1. 
This holds trivially if $\mathcal{D}^0=\varnothing$.
For $\mathcal{D}^0 \neq \varnothing$, applying Lemma~\ref{lem:reverse-picking} establishes that for any two agents $i,j\in N'$, either $v_i(X'_i)/w_i \ge v_i(X'_j)/w_j$ holds, or there exists some bundle $Q_t \subseteq X'_i$ received in iteration $t$ that satisfies
\begin{align*}
  \frac{v_i(X'_i \setminus Q_t)}{w_i} \ge \frac{v_i(X'_j)}{w_j}.
\end{align*}
Since $v_i(G) \ge 0$ for every $G \in \mathcal{G}_t$, and $e_t$ is the original chore contained in $Q_t$, we have
\begin{align*}
  \frac{v_i(X'_i \setminus \{e_t\})}{w_i}
  &= \frac{v_i(X'_i \setminus Q_t) + \sum_{G \in \mathcal{G}_t} v_i(G)}{w_i} 
  \ge \frac{v_i(X'_i \setminus Q_t)}{w_i} 
  \ge \frac{v_i(X_j')}{w_j}.
\end{align*}
Consequently, at the end of Step 1, the partial allocation $\mathbf{X}'$ satisfies the $\mathrm{WEF1}$ condition in the original instance for all agents in $N$.
\end{proof}

It remains to show that Step 2 preserves the $\mathrm{WEF1}$ condition.
\begin{lemma}\label{lem:leftover-metagoods}
The allocation $\bX$ returned by~\cref{alg:reduced-allocation} satisfies $\WEFone$.
\end{lemma}

\begin{proof}
We begin by considering the special case where $\mathcal{R} = \varnothing$. In this case, all meta-goods are completely allocated during Step 1, meaning Step 2 makes no further assignment. The claim then follows directly from \cref{lem:reverse-loop-structure}.

Therefore, it suffices to assume that $\mathcal{R} \neq \varnothing$.
We first claim that $N \setminus N' \neq \emptyset$. Otherwise, suppose that $N = N'$, this implies that every agent is involved in Step 1. 
By the definition of meta-goods, each meta-good has at least one agent with a non-negative valuation for it. 
Consequently, by the design of our algorithm, all meta-goods would be allocated in Step 1, contradicting the assumption that $\mathcal{R} \neq \varnothing$.
Thus, we must have $N \setminus N' \neq \varnothing$.

Next, we show that all remaining meta-goods in $\mathcal{R}$ are strictly negative for all agents in $N'$.
Consider any agent $i \in N'$. Since every $G \in \mathcal{R}$ has survived the entire reverse loop, it was available at the iteration $t$ where agent $i$ made their choice. 
If $v_i(G) \ge 0$, then $G$ would have been included in $\mathcal{G}_t$ by definition, contradicting $G \in \mathcal{R}$. 
Hence, we have
\begin{equation}\label{eq:leftover-negative-to-Nprime}
  v_i(G) < 0 \qquad \text{for all } i \in N' \text{ and } G \in \mathcal{R}.
\end{equation}

On the other hand, since every $G \in \mathcal{R}$ is a meta-good over $N$, it must be non-negatively valued by at least one agent in $N$. By \eqref{eq:leftover-negative-to-Nprime}, any such agent must belong to $N \setminus N'$. Hence, every $G \in \mathcal{R}$ is in fact a meta-good over $N \setminus N'$.
Since $N \setminus N' \neq \varnothing$, it naturally follows that $|\mathcal{D}^0| < n$, which ensures that the $\sigma(1), \ldots, \sigma(|\mathcal{D}^0|)$ are pairwise distinct i.e., each agent in $N'$ is assigned a unique receiving iteration. 
Furthermore, we have $X_i' = \varnothing$ for all $i \in N \setminus N'$.

Now we are ready to analyze the fairness guarantee of the allocation returned by Algorithm~\ref{alg:reduced-allocation}.
Fix any pair of agents $i, j \in N$. We present a case analysis based on their group memberships:

\begin{itemize}
    \item \textbf{Case 1: $i, j \in N'$}. 
    Neither agent receives any item during Step~2. Consequently, agent $i$ preserves $\mathrm{WEF1}$ towards $j$ directly by Lemma~\ref{lem:reverse-loop-structure}.

    \item \textbf{Case 2: $i, j \in N \setminus N'$}. 
    Since $X_i' = X_j' = \varnothing$, both agents receive items exclusively through the execution of $\mathrm{FWPS}(N \setminus N', \mathcal{R})$. Thus, the claimed fairness guarantee follows immediately from \cref{lem:forward-metagoods}.
    \item \textbf{Case 3: $i$ and $j$ belong to different groups}. 
    First, consider the subcase where $i \in N'$ and $j \in N \setminus N'$. Let $t$ be the unique iteration such that $\sigma(t) = i$. Then $X_i = X'_i = Q_t$ and $X_j = Y_j$. By \eqref{eq:leftover-negative-to-Nprime}, every meta-good in $Y_j$ yields strictly negative utility for agent $i$, which implies that $v_i(X_j) \le 0$. By removing the unique original chore $e_t$ from $X_i$, we have
    \begin{equation*}
      v_i(X_i \setminus \{e_t\}) = \sum_{G \in \mathcal{G}_t} v_i(G) \ge 0 \ge v_i(X_j),
    \end{equation*}
    which implies $v_i(X_i \setminus \{e_t\}) / w_i \ge v_i(X_j) / w_j$, and thus agent $i$ satisfies $\mathrm{WEF1}$ towards $j$.
    
    Conversely, consider the subcase where $i \in N \setminus N'$ and $j \in N'$. By \cref{lem:forward-metagoods}, we have $v_i(X_i) = v_i(Y_i) \ge 0$. Furthermore, by \cref{lem:construction}, we have $v_i(X_j) = v_i(Q_t) < 0$. It follows that $v_i(X_i) / w_i \ge 0 > v_i(X_j) / w_j$, so agent $i$ satisfies $\mathrm{WEF1}$ (in fact, $\mathrm{WEF}$) towards $j$.
\end{itemize}
Thus, agent $i$ satisfies $\mathrm{WEF1}$ towards $j$ in all cases, which completes the proof.
\end{proof}

\begin{lemma}\label{lem:reduced-allocation}
\Cref{alg:reduced-allocation} runs in $O(nm^2)$.
\end{lemma}
\begin{proof}
We analyze the time complexity phase by phase. 
In Step~1, the algorithm first determines the ordering $\sigma(1), \dots, \sigma(|\mathcal{D}^0|)$ by repeatedly selecting an agent $i \in N$ with the minimum $s_i$. 
Since $|\mathcal{D}^0| \le m$, scanning all $n$ agents at each step requires $O(nm)$ time in total. 
Next, the algorithm processes $\sigma(1), \dots, \sigma(|\mathcal{D}^0|)$ in reverse order for at most $m$ iterations. 
Within each iteration, $e_t$ and $\mathcal{G}_t$ are identified by scanning the remaining residual chores and meta-goods, respectively. 
Consequently, each iteration takes $O(m)$ time, yielding a total runtime of $O(m^2)$ for this step.
In Step~2, executing Algorithm~\ref{alg:forward-picking} involves up to $m$ iterations over at most $m$ meta-goods. In each iteration, identifying the agents in $N^0$ takes $O(nm)$ time by checking all pairs of agents and meta-goods.
Subsequently, selecting the agent $i$ with the minimum $w_i$ requires $O(n)$ time, while assigning her favorite remaining meta-good takes $O(m)$ time. 
Summing over all iterations yields a complexity of $O(nm^2)$ for this step.
Hence, \cref{alg:reduced-allocation} runs in $O(nm^2)$ time in total.
\end{proof}

Combining all these results together, we obtain the following theorem.
\begin{theorem}\label{thm:main}
Given any mixed manna instance, there exists an algorithm that computes a $\WEFone$ allocation in $O(nm^2\log m)$ time.
\end{theorem}

\begin{proof}
Given the mixed manna instance $\mathcal{I}$, we first apply \cref{alg:metagood-construction} to construct the meta-goods and residual chores $(\mathcal G,\mathcal D)$.
The resulting pair $(\mathcal G,\mathcal D)$ serves as a valid input to \cref{alg:reduced-allocation}, which outputs an allocation that is $\WEFone$ by \cref{lem:leftover-metagoods}.
Regarding computational complexity, \cref{lem:construction-runtime,lem:reduced-allocation} establish that the two steps take $O(nm^2\log m)$ time in total. 
Thus, the algorithm framework computes a $\WEFone$ allocation within the claimed running time.
This completes the proof.
\end{proof}

\section{The Existence of WEF1T and fPO Allocations}\label{sec:fpo}
We next study the compatibility of weighted fairness and fPO in mixed manna.
We first show that relaxing $\WEFone$ to $\WEFoneT$ is necessary when fPO is also required.
To state the impossibility result, we use the following fractional-transfer relaxation of WEF, which was introduced for goods by Chakraborty et al.~\cite{journals/teco/ChakrabortySS24}.

\begin{definition}[WEF$(x,y)$]
For $x,y\in[0,1]$, an allocation $\bX$ satisfies $\mathrm{WEF}(x,y)$ if, for any two agents $i,j\in N$, at least one of the following holds:
\begin{enumerate}[label=(\roman*),ref=(\roman*)]
    \item $v_i(X_i)/w_i\ge v_i(X_j)/w_j$;
    \item there exists $e\in X_j$ with $v_i(e)>0$ such that
$(v_i(X_i)+y \cdot v_i(e))/w_i\ge (v_i(X_j)-x \cdot v_i(e))/w_j$;
    \item there exists $e\in X_i$ with $v_i(e)<0$ such that
$(v_i(X_i)-x \cdot v_i(e))/w_i\ge (v_i(X_j)+y \cdot v_i(e))/w_j$.
\end{enumerate}
\end{definition}

Observe that by definition, setting $x=y=1$ yields WEF(1,1), which coincides precisely with $\WEFoneT$. 
On the other hand, setting $x=1$ and $y=0$ recovers the classical $\WEFone$ concept.
The following lemma shows that relaxing to one transfer is unavoidable when fPO is also required.

\begin{lemma}\label{thm:wef1-fpo-impossibility}
For any $x,y\in[0,1]$ with $x+y<2$, an allocation that satisfies both $\mathrm{WEF}(x,y)$ and fPO may not exist.
\end{lemma}

The proof of Lemma~\ref{thm:wef1-fpo-impossibility} is deferred to Appendix~\ref{app:wefxy-fpo-impossibility}.
Motivated by this impossibility result, we explore whether WEF1T can be achieved alongside fPO, which yields the following theorem.
\begin{theorem}\label{thm:wef1t-fpo}
Given any mixed manna instance, there exists an allocation that satisfies both $\WEFoneT$ and fPO.
\end{theorem}

The remainder of this section is devoted to proving \cref{thm:wef1t-fpo}.
The proof starts with a structural reduction that converts the original mixed instance into a reduced instance containing only objective goods and objective chores.
A similar reduction appears in~\cite{conf/ijcai/0001X025}.
We then augment the reduced instance by giving each agent a private item that only she values positively and adding universal zero dummies to equalize bundle sizes.
After this augmentation, we apply the Knaster--Kuratowski--Mazurkiewicz (KKM) lemma, following recent applications of the lemma in fair division~\cite{conf/wine/BarmanHSS25,journals/geb/IgarashiM26,journals/corr/abs-2509-18673,journals/corr/EF1andPO}.

The key existence step uses the KKM lemma~\cite{knaster1929beweis}, a fixed-point theorem for coverings of a simplex.
Informally, KKM says that if sets cover the faces of a simplex in the appropriate way, then all these sets have a common point.
For each agent, we define a set of weights for which some corresponding welfare-maximizing allocation satisfies the comparison required for that agent.
We show that these sets cover every face of the simplex, so KKM provides a weight vector lying in all of them.
A tie-breaking argument then combines the corresponding allocations into a single allocation that satisfies $\WEFoneT$, and a limiting argument establishes fPO.

\subsection{Reduction and Augmentation}
In this subsection, we perform a series of reduction and instance augmentation steps to simplify the problem structure.
This thereby allows us to restrict our attention to a well-behaved subclass of instances without loss of generality.

We begin by partitioning the items based on the sign of their maximum valuation across all agents. 
Let $M^0 = \{e \in M : \max_{k \in N} \{v_k(e)\} = 0\}$, and let $\bar{M} = M \setminus M^0$. 
Observe that it suffices to establish how to allocate items in $\bar{M}$ while guaranteeing $\mathrm{WEF1T}$. After doing so, each item $e \in M^0$ can be assigned to an agent $i$ with $v_i(e) = 0$ without violating the $\mathrm{WEF1T}$ condition.
For every agent $i\in N$ and item $e\in\bar M$, we further define the modified valuation
$$
\bar v_i(e)=
\begin{cases}
\max\{v_i(e),0\}, & \text{if }\max_{k\in N}\{v_k(e)\}>0,\\
v_i(e), & \text{if }\max_{k\in N}\{v_k(e)\}<0.
\end{cases}
$$

Note that under the modified valuation functions $\mathbf{\bar{v}}$, every item $e\in\bar M$ is either an objective good, with $\bar v_i(e)>0$ for some agent $i$, or an objective chore.
Intuitively, to achieve efficiency, an objective good should ideally be allocated to an agent $i$ who values it positively, i.e., $\bar{v}_i(e) > 0$. 

Guided by this principle, we formally introduce the notion of a non-wasteful allocation for the reduced instance $\mathcal{I'} = (N, \bar{M}, \mathbf{w}, \mathbf{\bar{v}})$. 
\begin{definition}[Non-wasteful Allocation]
An allocation $\bX' = (X_1', \dots, X_n')$ of the reduced instance $\mathcal{I'}$ is \emph{non-wasteful} if every item $e \in \bar{M}$ with $\max_{k \in N} \{\bar{v}_k(e)\} > 0$ is assigned to an agent $i \in N$ that satisfies $\bar{v}_i(e) > 0$.
\end{definition}

In what follows, to establish the existence of a $\mathrm{WEF1T}$ and $\mathrm{fPO}$ allocation, we show that it suffices to consider the reduced instance $\mathcal{I'}$.

\begin{lemma}\label{lem:reduction-lifting}
Let $\bX$ be obtained from a non-wasteful allocation $\bX'$ by allocating every item in $M^0$ to an agent who values it at zero.
If $\bX'$ satisfies $\WEFoneT$ and fPO for the reduced instance $\mathcal I'$, then $\bX$ satisfies $\WEFoneT$ and fPO for the original instance $\mathcal I$.
\end{lemma}
\begin{proof}
Fix any pair of agents $i, j \in N$. 
By non-wastefulness, we have $v_i(X_i) = \bar{v}_i(X_i')$.
In addition, by definition of the modified valuations, $v_i(e) \le \bar{v}_i(e)$ holds for every item $e \in X_j'$.
Consequently, under $\mathbf{v}$, agent $i$'s valuation of her own bundle remains identical to that under $\mathbf{\bar{v}}$, while her evaluation of agent $j$'s bundle does not increase.
Therefore, all WEF1T conditions under $\mathbf{\bar{v}}$ are thus preserved under $\mathbf{v}$. 
Moreover, allocating items in $M^0$ leaves each receiver's total valuation unchanged while only non-increasing other agents' valuations of those bundles. Thus, $\bX$ satisfies $\mathrm{WEF1T}$.

It remains to show that $\bX$ is $\mathrm{fPO}$.
Suppose for contradiction that there exists a fractional allocation $\bY$ that Pareto-dominates $\bX$ under $\mathbf{v}$. By discarding the fractional shares of items in $M^0$, we obtain a valid fractional allocation $\bY'$ of $\bar{M}$. For every agent $i$, their utility under $\mathbf{\bar{v}}$ satisfies $\bar{v}_i(Y_i') \ge v_i(Y_i)$, since items in $M^0$ have non-positive valuations for every agent and $\bar{v}_i(e) \ge v_i(e)$ holds for every item $e \in \bar{M}$. Because $v_i(Y_i) \ge v_i(X_i) = \bar{v}_i(X_i')$ for all agents $i$, with strict inequality for at least one agent, $\bY'$ Pareto-dominates $\bX'$ under $\mathbf{\bar{v}}$. This contradicts the assumption that $\bX'$ is $\mathrm{fPO}$ for the reduced instance $\mathcal{I}'$.
\end{proof}

Therefore, in the following, we focus on proving the existence of a $\mathrm{WEF1T}$ and $\mathrm{fPO}$ allocation for the reduced instance $\mathcal{I}'$. 
To this end, we augment the reduced instance with two types of auxiliary items: for each agent $i \in N$, we introduce a private item $e_i$ that is valued positively solely by agent $i$, along with a set of universal zero-dummy items that are valued at zero by all agents.
To specify the value of each private item, we first introduce the following definition. For any pair of agents $i, j \in N$, let
$$
\Gamma_{i,j}(\bX') =
\frac{\bar v_i(X_j')}{w_j} - \frac{\bar v_i(X_i')}{w_i}
- \left(\frac{1}{w_i} + \frac{1}{w_j}\right) \cdot
\max\Bigl\{
0,\
\max_{\substack{e \in X_j' \\ \bar v_i(e) > 0}} \{\bar v_i(e)\},\
\max_{\substack{e \in X_i' \\ \bar v_i(e) < 0}} \{-\bar v_i(e)\}
\Bigr\},
$$
where the maximum over an empty set is defined as $-\infty$. 
Intuitively, $\Gamma_{i, j}(\bX')$ measures the violation of $\mathrm{WEF1T}$ from $i$ to $j$.
In particular, $\bX'$ violates $\mathrm{WEF1T}$ from agent $i$ to agent $j$ if and only if $\Gamma_{i,j}(\bX') > 0$.
If every integral allocation satisfies $\mathrm{WEF1T}$, we simply set $\gamma = 1$. Otherwise, we define $\gamma > 0$ as the minimum value of $\Gamma_{i,j}(\bX')$ over all agents $i, j \in N$ and all integral allocations $\bX'$ violating $\mathrm{WEF1T}$, i.e.,
$$
\gamma = \min_{\bX' \in \prod_n(\bar M)} \left\{\min_{i, j \in N} \{\Gamma_{i, j}(\bX')\}\right\},
$$
where $\prod_n(\bar M)$ denotes all $n$-partition of $\bar M$ such that is not WEF1T.

We are now ready to augment the reduced instance $\mathcal{I}'$. 
Set $w_{\min} = \min_{i \in N} \{w_i\}$ and choose $0 < \delta < \gamma \cdot w_{\min}$. 
For each agent $i \in N$, we introduce a private item $e_i$ with valuation $\bar{v}_i(e_i) = \delta$ and $\bar{v}_j(e_i) = 0$ for all $j \neq i$. 
In addition, let $D^+$ be a set of $n(|\bar{M}| + n) - |\bar{M}|$ dummy items with $\bar{v}_i(e) = 0$ for all $i \in N$. 
We then define the augmented item set $M^+ = \bar{M} \cup \{e_i \mid i \in N\} \cup D^+$, which yields the resulting canonical instance $\mathcal{I}^+ = (N, M^+, \mathbf{w}, \mathbf{\bar{v}})$.

Intuitively, the private items ensure that every agent receives at least one item she uniquely values positively, while dummy items fill out the required number of items without altering any agent's utility.

In the following analysis, we temporarily restrict our attention to allocations where each agent receives exactly $|\bar{M}| + n + 1$ items. 
For simplicity, we refer to such an allocation as an \emph{equal-cardinality allocation}.
This restriction does not alter the achievable utility space of the reduced instance.
Indeed, once the original and private items are allocated, each agent $i \in N$ can be allocated with sufficient number of dummy items to meet the restriction without affecting $i$'s utility.
The following lemma shows that it suffices to prove the existence of an allocation that satisfies both $\mathrm{WEF1T}$ and $\mathrm{fPO}$ in the canonical instance $\mathcal{I}^+$.

\begin{lemma}\label{lem:canonical-projection}
Let $\bX^+$ be an allocation of $\mathcal{I}^+$ that assigns each private item $e_i$ to agent $i$, and let $\bX'$ be the allocation of $\mathcal{I}'$ obtained by discarding all auxiliary items from $\bX^+$.
If $\bX^+$ satisfies $\mathrm{WEF1T}$ and $\mathrm{fPO}$ for $\mathcal{I}^+$, then $\bX'$ guarantees both properties for $\mathcal{I}'$.
\end{lemma}
\begin{proof}
Fix any pair of agents $i, j \in N$. 
By construction, the private item $e_i$ increases agent $i$'s valuation of her own bundle by $\delta$, whereas all auxiliary items in agent $j$'s bundle have zero value to agent $i$. 
Moreover, none of these auxiliary items changes the maximum term in the definition of $\Gamma_{i,j}$. 
Consequently, if $\bX'$ were to violate $\mathrm{WEF1T}$ from agent $i$ to agent $j$, we would have $\Gamma_{i,j}(\bX') \ge \gamma$. This implies $\Gamma_{i,j}(\bX^+) \ge \gamma - \delta / w_i > 0$ as $\delta < \gamma \cdot w_{\min}$, which contradicts the assumption that $\bX^+$ satisfies $\mathrm{WEF1T}$. 
Thus, $\bX'$ must satisfy $\mathrm{WEF1T}$.

It remains to establish that $\mathbf{X}'$ guarantees $\mathrm{fPO}$. 
Suppose for contradiction that a fractional allocation $\mathbf{Y}'$ Pareto dominates $\mathbf{X}'$ in the reduced instance $\mathcal{I}'$. 
We can extend $\mathbf{Y}'$ to a fractional allocation $\mathbf{Y}^+$ of $\mathcal{I}^+$ by fully assigning each private item $e_i$ to agent $i$ and keeping the fractional dummy-item shares from $\mathbf{X}^+$. 
However, by construction, the allocation $\mathbf{Y}^+$ Pareto dominates $\mathbf{X}^+$, contradicting the $\mathrm{fPO}$ property of $\mathbf{X}^+$.
\end{proof}

The canonical instance $\mathcal{I}^+$ offers key structural advantages that enable the application of the KKM lemma, which we will introduce in Section~\ref{subsec: kkm}.

\subsection{Simplex Selection via KKM}\label{subsec: kkm}
In this subsection, we introduce the KKM lemma and apply it to establish the existence of allocation that ensures both $\mathrm{WEF1T}$ and $\mathrm{fPO}$ for the canonical instance $\mathcal{I}^+$. 
Let $\Delta^{n-1} = \{ (\theta_1,\dots,\theta_n) \in \mathbb{R}_{\ge 0}^n : \sum_{i \in N} \theta_i = 1 \}$ denote the $(n-1)$-dimensional standard simplex. For a parameter $\rho \in (0, 1/n)$, which we choose below, define $f : [0,1] \to \mathbb{R}_{>0}$ as
$$
  f(x) = \frac{1}{\rho + (1 - n\rho)x}.
$$
The function $f$ maps each coordinate in $[0, 1]$ to a strictly positive welfare weight. Since $f$ is strictly decreasing, smaller coordinates yield strictly larger welfare weights. In particular, we have $f(0) = 1/\rho$ and $f(1/n) = n$. 
This separation is used below to verify the boundary condition of the KKM lemma.
Fix $\mathbf{\theta} \in \Delta^{n-1}$ and consider the following primal-dual pair of linear programs:
$$
\renewcommand{\arraystretch}{1.25}
\begin{array}{@{}l@{\qquad\qquad\qquad}l@{}}
\textbf{LP}(\btheta,\rho)
& \textbf{Dual-LP}(\btheta,\rho) \\[1.5ex]
\begin{array}[t]{r@{\;}l}
\text{Max.} & \displaystyle\sum_{i\in N}\sum_{e\in M^+} f(\theta_i) \cdot \bar v_i(e) \cdot x_{i,e} \\[3 ex]
\text{s.t.} & \displaystyle\sum_{i\in N}x_{i,e}=1 \quad \forall e\in M^+,\\[3 ex]
& x_{i,e}\ge0 \quad \forall i\in N,\ \forall e\in M^+
\end{array}
&
\begin{array}[t]{r@{\;}l}
\text{Min.} & \displaystyle\sum_{e\in M^+} p(e) \\[3ex]
\text{s.t.} & p(e)\ge f(\theta_i) \cdot \bar v_i(e) \quad \forall i\in N,\ \forall e\in M^+
\end{array}
\end{array}
$$

Since the primal feasible region is the fractional allocation polytope, the linear program admits an integral optimal solution. 
Furthermore, because $f(\theta_i) > 0$ for all $i \in N$, any such integral optimum is necessarily $\mathrm{fPO}$.
The dual variables $p(e)$ for $e \in M^+$ are interpreted as item prices. Let $\mathbf p=(p(e))_{e\in M^+}$ denote the optimal dual price vector.
Since the dual constraints decompose itemwise, this price function is uniquely determined by
$$
   p(e) = \max_{i \in N} \left\{ f(\theta_i) \cdot \bar v_i(e) \right\}
  \qquad \text{for each } e \in M^+.
$$

For any non-empty $S \subseteq M^+$ with total price $p(S) = \sum_{e \in S} p(e)$, we define
$$
  \widehat p(S) = p(S) - \max_{e \in S} \{ p(e) \},
  \qquad
  \widetilde p(S) = p(S) - \min_{e \in S} \{ p(e) \}.
$$

Recall that we restrict our attention to allocations where every agent receives exactly $|\bar M| + n + 1$ items. That is, for any equal-cardinality allocation $\mathbf{X}^+$, we have $|X_i^+| = |\bar M| + n + 1$ for all $i \in N$.
We then define the tie-breaking potential function as
$$
\Phi(\mathbf{X}^+) = \sum_{i \in N} \sum_{\ell=1}^{|\bar M| + n}
\max \left\{ p(S) : S \subseteq X_i^+, \, |S| = \ell \right\}.
$$
For a given $\varepsilon > 0$, let $\mathcal{A}_\varepsilon$ denote the set of equal-cardinality allocations that maximize
$$
  \operatorname{Score}(\mathbf{X}^+)
  = \sum_{i \in N} f(\theta_i) \cdot \bar v_i(X_i^+) + \varepsilon \cdot \Phi(\mathbf{X}^+).
$$
Here, the primary term represents the weighted social welfare, while the second term is a perturbation that favors allocations with larger sums of expensive price prefixes.
The parameter $\varepsilon$ is chosen small enough so that this preference causes only a slight perturbation to the welfare.

Consider an equal-cardinality allocation that assigns each non-dummy item $e$ to an agent that maximizes $f(\theta_i)\cdot \bar v_i(e)$, which achieves the maximum weighted social welfare $\sum_{e\in M^+}p(e)$. Since $M^+$ is finite and $\rho$ is fixed, $\Phi(\mathbf{X}^+)$ is bounded by a constant independent of $\mathbf{\theta}$.
Therefore, any $\mathbf{X}^+\in\mathcal A_\varepsilon$ suffers at most an $O(\varepsilon)$ welfare loss relative to this maximum:
\begin{equation}\label{eq:price-welfare-gap}
0\le \sum_{e\in M^+}p(e)-\sum_{i\in N} f(\theta_i) \cdot \bar v_i(X_i^+)
 \le O(\varepsilon)
 \qquad \text{for all } \mathbf{X}^+\in\mathcal A_\varepsilon,
\end{equation}
where the lower bound follows from the definition of $p(e)$, and the $O(\varepsilon)$ bound holds uniformly over $\mathbf{\theta}$ and $\mathbf{X}^+$.
Next, we establish a structural property satisfied by allocations in $\mathcal{A}_\varepsilon$.
\begin{lemma}\label{lem:signed-interlacing}
Given $\mathbf{\theta} \in \Delta^{n-1}$, for any $\mathbf{X}^+, \mathbf{Y}^+ \in \mathcal{A}_\varepsilon$ and $i \in N$, we have $\widehat p(X_i^+)\le \widetilde p(Y_i^+)$.
\end{lemma}

\begin{proof}
To establish the claim, we construct a multigraph by taking one copy of every item from $\bX^+$ and one copy from $\bY^+$. 
For each agent, order her item copies in non-increasing order of price and pair consecutive copies. 
In addition, pair the two copies of each item. 
Every vertex has degree two, and the resulting multigraph therefore decomposes into disjoint even cycles whose edges alternate between the two types of pairs.

We now alternate the edges on every cycle, which yields two equal-cardinality allocations $\bX^{(0)}$ and $\bX^{(1)}$. 
Since each cycle merely exchanges item copies along alternating edges, the total weighted social welfare of $\bX^{(0)}$ and $\bX^{(1)}$ equals that of $\bX^+$ and $\bY^+$. Moreover, the pairing by non-increasing price ensures that, for every agent $i$ and every $\ell<|\bar M|+n+1$, the sum of the $\ell$ largest prices in $X_i^{(0)}$ and $X_i^{(1)}$ is at least the corresponding sum for $X_i^+$ and $Y_i^+$.

Consequently, we have $\Phi(\bX^{(0)})+\Phi(\bX^{(1)})\ge\Phi(\bX^+)+\Phi(\bY^+)$. Since $\bX^+$ and $\bY^+$ both maximize $\operatorname{Score}(\cdot)$, this inequality cannot be strict. Hence every prefix-sum comparison above holds with equality. Applying the equality for the relevant prefixes, together with the order in which the item copies are paired, it follows that $\widehat p(X_i^+)\le \widetilde p(Y_i^+)$.
\end{proof}

We are now ready to prove the existence of a $\mathrm{WEF1T}$ and $\mathrm{fPO}$ allocation for any canonical instance $\mathcal{I^+}$. 
To this end, we recall the classical KKM lemma~\cite{knaster1929beweis}, which serves as a set-covering analogue of Sperner's lemma.

\begin{lemma}[KKM Lemma~\cite{knaster1929beweis}]\label{lem:kkm}
Let $C_1,\ldots,C_n$ be closed subsets of $\Delta^{n-1}$.
Suppose that, for every $\btheta\in\Delta^{n-1}$, there is an agent $i\in N$ with $\theta_i>0$ and $\btheta\in C_i$.
Then $\bigcap_{i\in N}C_i\neq\varnothing$.
\end{lemma}

For every $\btheta \in \Delta^{n-1}$, let $p(\mathbf{\theta}, \cdot)$, $\Phi(\mathbf{\theta},\cdot)$, and $\mathcal{A}_\varepsilon(\mathbf{\theta})$ denote the price function, tie-breaking potential, and set of maximizers associated with $\mathbf{\theta}$ from the construction above. When $\mathbf{\theta}$ is clear from context, we drop this explicit dependence and write $p(\cdot)$, $\Phi(\cdot)$, and $\mathcal{A}_\varepsilon$, respectively.
Let $V = \max_{i \in N, e \in M} \{ |\bar{v}_i(e)| \}$ denote the maximum absolute valuation over the reduced original items, with $V = 0$ if no such item exists.
To verify the KKM boundary condition, it suffices to show that, at every point of each face of the simplex, every minimum label belongs to the support of that face. 
To this end, we choose $\rho$ sufficiently small so that
$$
0<\rho<\min\left\{\frac{1}{n},
 \frac{\delta}{(|\bar M|+n-1)\cdot nV+(n(|\bar M|+n+1)\cdot \max\{V,\delta\}+nV+1)/w_{\min}}\right\}.
$$

The central challenge lies in identifying an appropriate weight vector $\btheta$. 
Fix a parameter $\varepsilon>0$. 
For every $i\in N$, let $C_i\subseteq\Delta^{n-1}$ consist of the points $\btheta$ for which there exists $\bX^+\in\mathcal A_\varepsilon(\btheta)$ that satisfies $i\in\argmin_{k\in N}\{\widetilde p(X_k^+)/w_k\}$.

The following lemma establishes that these label sets are closed.
\begin{lemma}\label{lem:kkm-label-closed}
For every $i\in N$, the set $C_i$ is closed.
\end{lemma}

\begin{proof}
Let $\btheta^1,\btheta^2,\ldots\in C_i$ be a sequence converging to some $\btheta\in\Delta^{n-1}$.
For every $k\ge1$, we choose $\bX^k\in\mathcal A_\varepsilon(\btheta^k)$ for which $i$ attains the minimum in the definition of $C_i$.
Since there are only finitely many allocations, after passing to a subsequence, we may assume that $\bX^k=\bX^+$ for every $k$.
Since $\operatorname{Score}(\bX^+)$ is continuous in $\btheta$, $\bX^+$ also maximizes $\operatorname{Score}(\cdot)$ at $\btheta$.
Hence, $\bX^+\in\mathcal A_\varepsilon(\btheta)$.
Moreover, the price vector is continuous in $\btheta$, so $i$ remains a minimizer for $\bX^+$ at $\btheta$.
Thus, we have $\btheta\in C_i$, which implies that $C_i$ is closed.
\end{proof}

We now use these closed label sets to select the parameter $\btheta$.

\begin{lemma}\label{lem:kkm-selection}
There is an $\varepsilon_0>0$ such that, for every $\varepsilon\in(0,\varepsilon_0)$, we can choose $\btheta\in\Delta^{n-1}$ and $\bX^+\in\mathcal A_\varepsilon$ so that for any two agents $i,j\in N$,
$\widehat p(X_j^+)/{w_j} \le \widetilde p(X_i^+)/{w_i}$.
\end{lemma}

\begin{proof}
Since the upper bound in \eqref{eq:price-welfare-gap} is $O(\varepsilon)$ uniformly, there is an $\varepsilon_0>0$ such that it is smaller than $\min\{1,\delta/\rho\}$ for every $\varepsilon\in(0,\varepsilon_0)$.
Fix any $\varepsilon\in(0,\varepsilon_0)$ and consider the corresponding sets $C_i$ defined above.
By \cref{lem:kkm-label-closed}, every $C_i$ is closed.
In the following, we claim that
$$
 \Delta^{n-1}(N')\subseteq\bigcup_{i\in N'}C_i
 \quad\text{for every nonempty }N'\subseteq N,
$$
where $\Delta^{n-1}(N')=\{\btheta\in\Delta^{n-1}:\theta_i=0\text{ for }i\notin N'\}$.
Fix $\boldsymbol{\theta} \in \Delta^{n-1}(N')$ and let $\mathbf{X}^+$ be an allocation that maximizes $\operatorname{Score}(\cdot)$.
Since $\sum_{i \in N'} \theta_i = 1$, there exists some agent $j \in N'$ with $\theta_j \ge 1/n$, which implies $f(\theta_j) \le n$. 
Consequently, every objective chore $e$ receives a price that satisfies  $p(e) \ge f(\theta_j) \cdot \bar{v}_j(e) \ge -nV$, whereas each remaining non-private item $e$ have non-negative prices $p(e) \ge f(\theta_i) \cdot \bar{v}_i(e) \ge 0$.
On the other hand, for $i\notin N'$, we have $\theta_i = 0 $ and $f(\theta_i)=1/\rho$. 

Note that the welfare loss in \eqref{eq:price-welfare-gap} is smaller than $\delta/\rho$, whereas assigning $e_i$ to any agent other than $i$ contributes exactly $\delta/\rho$ to that loss. 
Hence $e_i$ is assigned to $i$. 
Since its price is $\delta/\rho$ and $X_i^+$ contains a zero dummy item, the private item $e_i$ cannot be the minimum-price item. 
Consequently, every remaining item, with the possible exception of the minimum-price item, has a price bounded below by $-nV$.
It follows that 
$$
\widetilde p(X_i^+)/w_i\ge \bigl(\delta/\rho-(|\bar M|+n-1)nV\bigr)/w_i\ge \delta/\rho-(|\bar M|+n-1)nV,$$
where the last inequality holds by $w_i\le1$.

For the upper bound for $j$, the part of the welfare loss in \eqref{eq:price-welfare-gap} contributed by $X_j^+$ is less than $1$, so $\sum_{e\in X_j^+}\bigl(p(e)-f(\theta_j) \cdot \bar v_j(e)\bigr)<1$.
The bundle $X_j^+$ has exactly $|\bar M|+n+1$ items, each worth at most $\max\{V,\delta\}$ to $j$, and hence $\bar v_j(X_j^+)\le (|\bar M|+n+1)\max\{V,\delta\}$. Moreover, $\min_{e\in X_j^+}p(e)\ge-nV$: for an objective chore, the dual constraint for $j$ gives $p(e)\ge f(\theta_j)\bar v_j(e)\ge-nV$, while every other item has nonnegative price. Since $f(\theta_j)\le n$, these observations give 
$$
p(X_j^+)\le n(|\bar M|+n+1) \cdot \max\{V,\delta\}+1.
$$ 
After subtracting the minimum item price and using $w_j\ge w_{\min}$, we obtain 
$$
\widetilde p(X_j^+)/w_j\le \bigl(n(|\bar M|+n+1)\max\{V,\delta\}+nV+1\bigr)/w_{\min}.
$$
By the choice of $\rho$, the lower bound above exceeds this upper bound. Thus no agent outside $N'$ is a minimum label, proving the claimed face condition.
In particular, every $\btheta\in\Delta^{n-1}$ belongs to some $C_i$ with $\theta_i>0$.
By \cref{lem:kkm}, there is a point $\btheta\in\bigcap_iC_i$.
Fix such a point $\btheta$. For the remainder of the proof, write $\mathbf p=p(\btheta,\cdot)$ and $\mathcal A_\varepsilon=\mathcal A_\varepsilon(\btheta)$.

Now we choose $\bX^+\in\mathcal A_\varepsilon$ that maximizes
$\min_{k\in N}\{\widetilde p(X_k^+)/w_k\}$ over all allocations in $\mathcal A_\varepsilon$.
Fix any $j\in N$. Since $\btheta\in C_j$, there exists an allocation $\bY^+\in\mathcal A_\varepsilon$ such that $j\in\argmin_{k\in N}\{\widetilde p(Y_k^+)/w_k\}$.
Applying \cref{lem:signed-interlacing}, it follows that for every $i \in N$, 
\begin{equation*}
 \frac{\widehat p(X_j^+)}{w_j}
 \le \frac{\widetilde p(Y_j^+)}{w_j}
 \le \min_{k\in N}\left\{\frac{\widetilde p(X_k^+)}{w_k}\right\}
 \le \frac{\widetilde p(X_i^+)}{w_i}. \qedhere
\end{equation*}
\end{proof}

\subsection{The Limiting Argument}
In this subsection, we assemble the established lemmas to complete the proof of Theorem~\ref{thm:wef1t-fpo} via a limiting argument.

We first construct an allocation that satisfies WEF1T and fPO of the canonical instance $\mathcal I^+$.
Once we obtain an allocation $\bX^+$ that is $\WEFoneT$ and fPO and assigns every private item $e_i$ to agent $i$, \cref{lem:canonical-projection} yields a $\WEFoneT$ and fPO allocation $\bX'$ of the reduced instance $\mathcal I'$.
We next show that $\bX'$ is non-wasteful.
We can then assign every item in $M^0$ to an agent who values it at zero and apply \cref{lem:reduction-lifting} to obtain the desired allocation of the original instance $\mathcal I$.

Let $\varepsilon_1,\varepsilon_2,\ldots\in(0,\varepsilon_0)$ be a sequence converging to $0$, where $\varepsilon_0$ is given by \cref{lem:kkm-selection}.
For every $t$, that lemma provides a point $\btheta^t\in\Delta^{n-1}$ and an allocation $\bX^t\in\mathcal A_{\varepsilon_t}(\btheta^t)$ satisfying its price comparison; write $\mathbf p^t=p(\btheta^t,\cdot)$ for the associated price vector.
By finiteness of the set of allocations and compactness of $\Delta^{n-1}$, after passing to a subsequence, we may assume that $\bX^t=\bX^+$ for every $t$ and that $\btheta^t\to\btheta^*$.
Set $\mathbf p^*=p(\btheta^*,\cdot)$.

Because $\bX^+\in\mathcal A_{\varepsilon_t}(\btheta^t)$ for every $t$, the uniform bound in \eqref{eq:price-welfare-gap} gives
$$
0\le \sum_{e\in M^+}p^t(e)-\sum_{i\in N}f(\theta_i^t) \cdot\bar v_i(X_i^+)\le O(\varepsilon_t).
$$
For every item $e\in X_i^+$, dual feasibility gives $p^t(e)\ge f(\theta_i^t)\cdot \bar v_i(e)$, and the sum of these nonnegative slacks is the expression above.
By the continuity of the price function, each assigned edge becomes tight in the limit, i.e., 
$p^*(e)=f(\theta_i^*) \cdot \bar v_i(e)$ for every $i\in N$ and $e\in X_i^+$.
Thus, $\bX^+$ is primal feasible for $\textbf{LP}(\btheta^*,\rho)$, $\mathbf p^*$ is dual feasible for $\textbf{Dual-LP}(\btheta^*,\rho)$, and their objective values agree.
By strong duality, $\bX^+$ maximizes $\sum_i f(\theta_i^*) \cdot \bar v_i(X_i^+)$ over all fractional allocations.
Since all weights $f(\theta_i^*)$ are strictly positive, $\bX^+$ is fPO for $\mathcal I^+$.

The price comparisons guaranteed by \cref{lem:kkm-selection} also pass to the limit: for every $i,j\in N$, $\widehat p^*(X_j^+)/w_j\le \widetilde p^*(X_i^+)/w_i$.
Fix any two agents $i,j\in N$.
Let $c\in X_i^+\cup\{\varnothing\}$ minimize $\bar v_i$ and let $g\in X_j^+\cup\{\varnothing\}$ maximize $\bar v_i$.
Thus, $\bar v_i(c)\le0\le \bar v_i(g)$.
The tightness of the assigned edges in $X_i^+$ gives $\widetilde p^*(X_i^+)=f(\theta_i^*)\cdot\bar v_i(X_i^+\setminus\{c\})$, while dual feasibility on $X_j^+$ gives $f(\theta_i^*)\cdot\bar v_i(X_j^+\setminus\{g\})\le\widehat p^*(X_j^+)$.
Indeed, for the second inequality, remove an item of maximum price from $X_j^+$; its value to $i$ is no larger than that of $g$, and dual feasibility applies to every remaining item.
Combining these two bounds with the limiting comparison yields
\begin{equation}\label{eq:two-deletions}
 \frac{\bar v_i(X_i^+\setminus\{c\})}{w_i}
 \ge
 \frac{\bar v_i(X_j^+\setminus\{g\})}{w_j}.
\end{equation}
In what follows, we establish that $\mathbf{X}^+$ satisfies $\WEFoneT$ via case analysis.
\begin{itemize}
    \item \textbf{Case 1: $\bar{v}_i(c) = \bar{v}_i(g) = 0$.} Note that this naturally includes the case $c = g = \emptyset$. Here, Inequality~\eqref{eq:two-deletions} directly implies that agent $i$ ensures WEF towards agent $j$.
    
    \item \textbf{Case 2: $\bar{v}_i(g) > 0$ and $\bar{v}_i(g) \ge -\bar{v}_i(c)$.} 
    In this case, it holds that
    $$
      \bar{v}_i(X_i^+ \cup \{g\}) = \bar{v}_i(X_i^+) + \bar{v}_i(g)
      \ge \bar{v}_i(X_i^+) - \bar{v}_i(c) = \bar{v}_i(X_i^+ \setminus \{c\}).
    $$
    Combining this with \eqref{eq:two-deletions} yields
    $\bar{v}_i(X_i^+ \cup \{g\}) / {w_i} \ge {\bar{v}_i(X_j^+ \setminus \{g\})} / {w_j}$.
    
    \item \textbf{Case 3:} $-\bar{v}_i(c) > \bar{v}_i(g)$. This condition implies that $c \neq \emptyset$ and $\bar{v}_i(c) < 0$. Consequently,
    $$
      \bar{v}_i(X_j^+ \cup \{c\}) = \bar{v}_i(X_j^+) + \bar{v}_i(c)
      < \bar{v}_i(X_j^+) - \bar{v}_i(g) = \bar{v}_i(X_j^+ \setminus \{g\}).
    $$
    Together with \eqref{eq:two-deletions}, it follows that
    ${\bar{v}_i(X_i^+ \setminus \{c\})}/{w_i} \ge {\bar{v}_i(X_j^+ \cup \{c\})}/{w_j}$.
\end{itemize}
Hence, we conclude that $\mathbf{X}^+$ guarantees $\mathrm{WEF1T}$ for the canonical instance $\mathcal{I}^+$.

The tightness condition also implies that every private item $e_i$ is assigned to agent $i$, since $i$ uniquely maximizes $f(\theta_k^*)\cdot\bar v_k(e_i)$ over all $k\in N$.
Let $\bX'$ be obtained from $\bX^+$ by removing all private items and dummy items.
Since $\bX^+$ is $\WEFoneT$ and fPO and assigns each $e_i$ to agent $i$, \cref{lem:canonical-projection} implies that $\bX'$ is $\WEFoneT$ and fPO for $\mathcal I'$.
Moreover, consider any $e\in\bar M$ with $\max_{k\in N}\{v_k(e)\}>0$.
Then $\max_{k\in N}\{f(\theta_k^*)\bar v_k(e)\}>0$; the tightness of the edge assigning $e$ therefore implies that its owner values $e$ positively under both $\bar v$ and $v$.
Hence $\bX'$ is non-wasteful.
Finally, assign every item in $M^0$ to an agent who values it at zero.
Applying \cref{lem:reduction-lifting} to the resulting allocation of $\mathcal I$ proves that it is $\WEFoneT$ and fPO.

\section{Best-of-both-worlds Guarantee for Weighted Envy-freeness}\label{sec:bobw}
In this section, we shift our attention to randomized allocations and investigate the existence of Bobw guarantees. 
Specifically, we focus on ex-ante WEF and explore its compatible ex-post guarantees. 
Hoefer et al.~\cite{journals/jair/HoeferSV24} showed that, even in the goods-only setting, ex-ante WEF and ex-post $\mathrm{WEF}(x, y)$ are incompatible whenever $x + y < 2$ for $x, y \in [0, 1]$.
Since the goods setting constitutes a special case of mixed manna, this impossibility result extends immediately to our setting.
\begin{lemma} [\cite{journals/jair/HoeferSV24}]\label{lemma: bobw_impossibility}
    For any $x, y \in [0, 1]$ with $x +y < 2$, ex-ante WEF and ex-post WEF($x, y$) are incompatible.
\end{lemma}
Given the impossibility result introduced in Lemma~\ref{lemma: bobw_impossibility}, a natural question is whether there exists a randomized allocation that guarantees ex-ante $\mathrm{WEF}$ and ex-post $\mathrm{WEF1T}$.

A finite lottery over allocations is a family $\mathcal L=\{(p^t,\bX^t):t\in[T]\}$, where $p^t>0$ is the probability assigned to allocation $\bX^t$ and $\sum_{t\in[T]}p^t=1$.
The support of $\mathcal L$ is the set $\{\bX^t:t\in[T]\}$ of allocations with positive probability.
Let $\mathbf{X} \sim \mathcal{L}$ be a random allocation drawn from a lottery $\mathcal{L}$. 
The distribution $\mathcal{L}$ is \emph{ex-ante WEF} if, for all $i, j \in N$,
\begin{equation*}
\mathbb{E}_{\bX\sim\mathcal L}\left[\frac{v_i(X_i)}{w_i}\right]
\ge
\mathbb{E}_{\bX\sim\mathcal L}\left[\frac{v_i(X_j)}{w_j}\right].
\end{equation*}
When $\mathcal{L}$ is a finite lottery over allocations $\{ \mathbf{X}^t \}_{t \in [T]}$ with probabilities $p^t$, 
this is equivalently expressed as
\begin{equation*}
\sum_{t\in[T]}p^t\frac{v_i(X_i^t)}{w_i}
\ge
\sum_{t\in[T]}p^t\frac{v_i(X_j^t)}{w_j}.
\end{equation*}
Furthermore, $\mathcal{L}$ is \emph{ex-post $\mathrm{WEF1T}$} if every allocation in its support satisfies $\mathrm{WEF1T}$. 
For any coalition $S \subseteq N$, we write $w(S) = \sum_{i \in S} w_i$.

Our main finding is formalized in the following theorem.
\begin{theorem}\label{thm:bobw}
Given any mixed manna instance, there exists a lottery that is ex-ante WEF and ex-post $\WEFoneT$.
\end{theorem}
The remainder of this section is devoted to the proof of Theorem~\ref{thm:bobw}.

\subsection{The Construction of Lottery}

The proof reuses the meta-goods constructed by \cref{alg:metagood-construction}.
Let $(\mathcal G,\mathcal D)$ be its output, and write $P_G=\{i\in N:v_i(G)>0\}$ for every $G\in\mathcal G$.
Let $\mathcal G^0=\{G\in\mathcal G:P_G=\varnothing\}$ and $\mathcal G^+=\mathcal G\setminus\mathcal G^0$.
For each $G\in\mathcal G^0$, fix an agent who values $G$ at zero and assign $G$ to this agent deterministically.
Such a meta-good has zero value to its owner and nonpositive value to every other agent.
We henceforth construct the lottery only for the meta-goods in $\mathcal G^+$.
Moreover, for every $i\in P_G$, condition (ii) in \cref{def:metagood} supplies an item $g_i(G)\in G$ such that
\begin{equation}\label{eq:bobw-witness}
 v_i(g_i(G))\ge v_i(G)>0
 \qquad\text{and}\qquad
 v_i(G\setminus\{g_i(G)\})\le0.
\end{equation}

\begin{lemma}\label{lem:bobw-subjective-goods}
Let $N'\subseteq N$ be nonempty, and let $\mathcal G'\subseteq\mathcal G^+$ be a subfamily such that every $G\in\mathcal G'$ is strictly positive to some agent in $N'$.
There is a finite lottery over allocations of $\mathcal G'$ among $N'$ that is ex-ante WEF and ex-post $\WEFoneT$.
Moreover, every $G\in\mathcal G'$ is assigned only to an agent who values it strictly positively.
\end{lemma}

\begin{proof}
Let $\delta=\min_{i\in N'}\{w_i\}$.
For the purpose of applying the goods decomposition, define the truncated valuations by $v_i^+(G)=\max\{v_i(G),0\}$.
We further add more than $|\mathcal G'|(w(N')-\delta)/\delta$ dummy items that have value zero to all agents in $N'$, and we rank every dummy item above every real object with zero truncated value.
We use the different-speeds eating (DSE) procedure on the real and dummy objects under the truncated valuations.
At every time, agent $i$ continuously consumes an available object maximizing $v_i^+$ at speed $w_i$ and moves to a new favorite object whenever the current object is exhausted, with ties broken consistently.

The dummy objects ensure that this truncation does not change the DSE trajectory, since an agent prefers a dummy to every real object with nonpositive original value.
Let $x_{i,G}$ denote the fraction of object $G$ consumed by agent $i$ under this procedure.
While a real object remains, at least one agent consumes a real object at speed at least $\delta$.
Thus all real mass is exhausted within time $|\mathcal G'|/\delta$.
During this time, the total consumed mass is at most $w(N')|\mathcal G'|/\delta$, so at most $|\mathcal G'|(w(N')-\delta)/\delta$ dummy objects can be consumed.
A dummy therefore remains until every real object is exhausted, and the DSE matrix $x$ satisfies $x_{i,G}=0$ whenever the original value $v_i(G)$ is nonpositive.

Apply the exact hierarchy-and-quota decomposition for DSE on goods from~\cite{journals/jair/HoeferSV24}.
It is ex-ante WEF, every outcome is $\WEFoneT$, and its singleton quotas preserve the zero cells of $x$.
Consequently, every real object is assigned only on a strictly positive edge.
Remove the dummies and restore the negative values.
No agent receives an object whose value was changed, so every own value is unchanged, and every target value can only decrease.
Both the expected and outcome-wise inequalities are therefore preserved.
\end{proof}

We first consider the case in which no residual chore remains.
If $\mathcal D=\varnothing$, apply \cref{lem:bobw-subjective-goods} to $\mathcal G^+$ over $N$, add the fixed assignments of $\mathcal G^0$, and expand every meta-good to its constituent items.
This already gives the asserted lottery.
We now consider the case in which $\mathcal D\neq\varnothing$.
For the residual chores, we start from the DSE-based chores-only BoBW lottery of Theorem 4.13 in Wu et al.~\cite{journals/ai/WuZZ25}.
Their result supplies the ex-ante WEF and ex-post $\WEFoneT$ guarantees for chores.
Our mixed-instance proof additionally needs to couple the residual-chore lottery with meta-good ownership, so we retain the structural properties of the DSE decomposition that enable this coupling.
Every residual chore has negative value to every agent.
For residual chores, we use the analogous value version of DSE: agent $i$ continuously consumes an available chore with largest value to $i$ at speed $w_i$ and moves to a new largest-value chore whenever the current chore is exhausted, with ties broken consistently.
Let $x_{i,e}$ be the resulting fraction.
Then
\begin{equation*}
 \sum_{i\in N}x_{i,e}=1\quad\text{for all }e\in\mathcal D,
 \qquad
 \sum_{e\in\mathcal D}x_{i,e}=|\mathcal D|w_i\quad\text{for all }i\in N.
\end{equation*}
The continuous comparison underlying the chores-only BoBW construction of Wu et al.~\cite{journals/ai/WuZZ25} shows that the fractional allocation $x$ is WEF.

Write $|\mathcal D|w_i=a_i+f_i$, where $a_i=\lfloor |\mathcal D|w_i\rfloor$ and $0\le f_i<1$.
For every agent $i$, order the residual chores from worst to best as $e^i_1,\ldots,e^i_{|\mathcal D|}$, using a fixed refinement of ties, and put $s_{i,k}=\sum_{r\le k}x_{i,e^i_r}$.
Introduce $a_i$ mandatory slots and, when $f_i>0$, one optional slot for agent $i$.
The edge between $e^i_k$ and slot $(i,r)$ has weight
\begin{equation*}
 z_{e^i_k,(i,r)}=\bigl|[s_{i,k-1},s_{i,k})\cap[r-1,r)\bigr|.
\end{equation*}
We call a slot matching feasible if it uses only edges with positive weight.
Every chore and every mandatory slot has degree one, whereas optional slot $i$ has degree $f_i$.
Consequently, $z$ lies in a bipartite partial-matching polytope, whose feasible integral points are precisely the feasible slot matchings.

The following lemma records the part of the chores-only BoBW construction that we use, written in our notation.
\begin{lemma}[\cite{journals/ai/WuZZ25}]\label{lem:bobw-chore-law}
Let $\mu$ be a distribution over feasible integral slot matchings whose expected incidence vector is $z$.
If $\bY=(Y_i)_{i\in N}$ is the residual-chore allocation induced by a matching in the support of $\mu$, then
\begin{enumerate}[label=(\alph*),ref=(\alph*)]
\item $\Pr_{\mu}[e\in Y_i]=x_{i,e}$ for every $i\in N$ and $e\in\mathcal D$;
\item every $e\in Y_i$ satisfies $x_{i,e}>0$, and, for every $i\in N$ and $k\in[|\mathcal D|]$,
\begin{equation*}
 \lfloor s_{i,k}\rfloor
 \le |Y_i\cap\{e^i_1,\ldots,e^i_k\}|
 \le \lceil s_{i,k}\rceil;
\end{equation*}
\item if $Y_i\neq\varnothing$ and $e_i$ is $i$'s worst chore in $Y_i$, then, for every $j\in N$ and every valuation function $\widehat v_i$ that is negative on $\mathcal D$ and has the same tie-refined order as $v_i$,
\begin{equation}\label{eq:bobw-ordinal-chore}
 \frac{\widehat v_i(Y_i\setminus\{e_i\})}{w_i}
 \ge \frac{\widehat v_i(Y_j\cup\{e_i\})}{w_j};
\end{equation}
\end{enumerate}
\end{lemma}

We next select, among the distributions satisfying the assumptions of \cref{lem:bobw-chore-law}, one with the additional coupling properties needed below.
\begin{lemma}\label{lem:bobw-chore-coupling}
There is a distribution $\mu$ over feasible integral slot matchings whose expected incidence vector is $z$.
For the residual-chore allocation $\bY$ induced by a matching in the support of $\mu$, define $B(\bY)=\{i:|Y_i|=a_i+1\}$.
The distribution additionally satisfies
\begin{equation*}
 |B(\bY)|=\sum_i f_i
 \qquad\text{and}\qquad
 \Pr_{\mu}[i\in B(\bY)]=f_i.
\end{equation*}
Moreover, for disjoint $J,L\subseteq N$, with $f(S)=\sum_{i\in S}f_i$, if $f(J)+f(L)=0$, then $\Pr[B(\bY)\cap L\neq\varnothing]=0$.
Otherwise,
\begin{alignat}{1}
 \Pr[B(\bY)\cap J=\varnothing,\ B(\bY)\cap L\neq\varnothing]
 &\le\frac{f(L)}{f(J)+f(L)},\label{eq:bobw-paired-one}\\
 \Pr\!\left[B(\bY)\cap J\neq\varnothing\right]
 &\ge\frac{f(J)}{1+f(J)}.\label{eq:bobw-paired-two}
\end{alignat}
\end{lemma}

\begin{proof}
Let $\Omega$ be the finite set of feasible integral matchings.
The point $z$ lies in the relative interior of $\operatorname{conv} \{\chi^K:K\in\Omega\}$: every retained edge has positive $z$-value, every optional load lies strictly between zero and one, and the remaining equalities define the affine hull.
Maximize Shannon entropy over distributions on $\Omega$ with mean $z$.
The maximizer has full support and, by the Lagrange-multiplier conditions, has the edge-product form
\begin{equation}\label{eq:bobw-edge-product}
 \mu(K)=Z^{-1}\exp\!\left(\sum_{a\in K}\theta_a\right),
 \qquad \mathbb E_{\mu}[\chi^K]=z.
\end{equation}
By \cref{lem:bobw-chore-law}, these properties hold automatically.
The slot degrees and the mean constraint also give $|B(\bY)|=\sum_i f_i$ and $\Pr_{\mu}[i\in B(\bY)]=f_i$.

It remains to prove the two probability bounds.
Draw $K,K'$ independently from $\mu$, with induced allocations $\bY$ and $\bY'$, and write $B=B(\bY)$ and $B'=B(\bY')$.
By \eqref{eq:bobw-edge-product}, this conditioning makes the orientations of all alternating paths and cycles independent fair coins.
The path endpoints are optional slots.
Hence, at the level of $(B,B')$, there is a common set and disjoint endpoint pairs, with exactly one endpoint of each pair in $B$.
It remains to verify the following inequality on every orientation orbit:
$$
 \mathbb E\bigl[\mathbf 1_{\{B\cap J=\varnothing,\,B\cap L\neq\varnothing\}}
 |B'\cap(J\cup L)|\bigr]\le\mathbb E[|B'\cap L|].
$$
Fix an orientation orbit.
Let $h$ be the number of common $L$-elements and let $a,b,d,e$ be the numbers of endpoint pairs of types $O$--$J$, $O$--$L$, $J$--$L$, and $L$--$L$, respectively, where $O=N\setminus(J\cup L)$.
Common $J$-elements and $J$--$J$ pairs make the event impossible.
Otherwise,
\(\mathbb E[|B'\cap L|]=h+b/2+d/2+e\).
If $h+d+e>0$, then missing $J$ fixes the orientations of the $a+d$ pairs that meet $J$ and already guarantees that $B$ meets $L$.
Thus, the left-hand side equals \(2^{-(a+d)}(h+a+d+e+b/2)\).
If $d\ge1$, this is at most
\(h+e+b/2+d/2\), where we use \(a+d\ge1\) and \((a+d)2^{-(a+d)}\le d/2\).
If $d=0$, then $h+e>0$ and hence $h+e+b/2\ge1$.
The inequality \(a\le(2^a-1)(h+e+b/2)\) gives the same upper bound.
It remains to consider $h=d=e=0$.
If $b=0$, the event is impossible.
For $b\ge1$, conditional on missing $J$, let $U\sim\operatorname{Bin}(b,1/2)$ count the $O$--$L$ pairs whose $O$-endpoint enters $B$.
The event requires \(U<b\), and the left-hand side is \(2^{-a}(a+b/2-(a+b)2^{-b})\), which is at most \(b/2\).
Indeed, for \(b=1\) it equals \(a/2^{a+1}\le1/2\), while for \(b\ge2\) and \(a\ge1\) it is at most \(a/2^a+b/2^{a+1}\le1/2+b/4\le b/2\), and for \(a=0\) it equals \(b(1-2^{1-b})/2\le b/2\).
This proves the displayed orbit inequality.
Averaging it and using independence of $B,B'$ gives the product form of the first bound, which yields \eqref{eq:bobw-paired-one} after division by $f(J)+f(L)>0$.

For \eqref{eq:bobw-paired-two}, common $J$-elements and $J$--$J$ pairs again make missing $J$ impossible.
Otherwise let $r$ be the number of endpoint pairs with exactly one endpoint in $J$.
Missing $J$ has probability $2^{-r}$, and then $B'$ contains all $r$ corresponding $J$-endpoints.
Hence $\mathbb E[\mathbf 1_{\{B\cap J=\varnothing\}}|B'\cap J|] \le 2^{-r}r\le1-2^{-r}=\Pr[B\cap J\neq\varnothing]$.
Averaging proves \eqref{eq:bobw-paired-two}.
\end{proof}

Put $H=\{i:a_i\ge1\}$ and $L=N\setminus H$.
An agent in $L$ receives either zero or one residual chore.
For an outcome $\bY$, write $B=B(\bY)$ and, for every nonempty $T\subseteq N$ with $T\cap H\neq\varnothing$, define the safe agents by
\begin{equation*}
 \mathcal S_T(B)=\left\{j\in T:q_j>0\ \text{and}\
 \frac{q_i-1}{w_i}\le\frac{q_j}{w_j}\ \text{for all }i\in T\right\},
 \qquad q_i=a_i+\mathbf 1_{\{i\in B\}}.
\end{equation*}

The set $\mathcal S_T(B)$ identifies the agents who are admissible recipients for a meta-good with positive set $T$.
After the residual-chore lottery fixes $B$, we perform an additional random choice over $\mathcal S_T(B)$.

\begin{lemma}\label{lem:bobw-safe-owner}
For every nonempty $T\subseteq N$ with $T\cap H\neq\varnothing$, there is a family of probability distributions $\{\pi_T^B\}_B$ indexed by the sets $B$ that can arise as $B(\bY)$ under $\mu$.
For each such $B$, $\pi_T^B$ is a probability distribution over $\mathcal S_T(B)$.
We write $\pi_T^B(j)$ for the probability that agent $j$ is selected, and set $\pi_T^B(j)=0$ for $j\notin\mathcal S_T(B)$.
Here the sum over $B$ ranges over the sets that can arise as $B(\bY)$ under $\mu$.
These distributions satisfy
\begin{equation}\label{eq:bobw-safe-marginal}
 \sum_B\Bigl(\Pr_{\mu}[B]\pi_T^B(j)\Bigr)=\frac{w_j}{\sum_{i\in T}w_i}
 \qquad(j\in T).
\end{equation}
\end{lemma}

\begin{proof}
Let $\rho_i=f_i/w_i$.
For $J\subseteq T$, we prove $\Pr[\mathcal S_T(B)\cap J\neq\varnothing]\ge w(J)/w(T)$.
If $J$ contains an agent that is active in the sense $a_j>0$ and $\rho_j\le\min_{i\in T}\{\rho_i+1/w_i\}$, let $L_0$ be the agents in $T$ of strictly smaller $\rho$-value.
Failure to meet $J$ is exactly the event that $B$ misses $J$ and meets $L_0$.
Thus \eqref{eq:bobw-paired-one} and $f(L_0)/w(L_0)<f(J)/w(J)$ give
$$
 \Pr[\mathcal S_T(B)\cap J=\varnothing]
 \le\frac{f(L_0)}{f(J)+f(L_0)}
 \le1-\frac{w(J)}{w(T)}.
$$
If $J$ contains no active agent, an active $k\in T\setminus J$ attaining the minimum satisfies $f(J)\ge w(J)(f_k+1)/w_k\ge w(J)$.
Since $w(T)-w(J)\ge w_k$, \eqref{eq:bobw-paired-two} yields $\Pr[B\cap J\neq\varnothing]\ge w(J)/w(T)$; hitting $J$ is then sufficient for safety.
The cases $J=\varnothing$ and $f_i=0$ for all $i$ are immediate.

The displayed inequality is precisely the fractional Hall condition for the bipartite graph connecting each outcome $B$ to the agents in $\mathcal S_T(B)$.
Fractional Hall, equivalently max-flow/min-cut, gives the family of distributions $\{\pi_T^B\}_B$ with marginals \eqref{eq:bobw-safe-marginal}.
\end{proof}

\subsection{Combining the Components}

Now we are ready to prove Theorem~\ref{thm:bobw}.
Sample a residual matching from $\mu$, obtaining $\bY$, and let $B=B(\bY)$ and $q_i=a_i+\mathbf 1_{\{i\in B\}}$ for every $i\in N$.
For every $G\in\mathcal G^+$, proceed independently conditional on $B$ as follows.
If $P_G\cap H\neq\varnothing$, choose its owner from $P_G$ according to $\pi_{P_G}^B$.
If $P_G\subseteq L$, allocate all such meta-goods jointly by \cref{lem:bobw-subjective-goods} on $L$.
Finally add the fixed assignments of $\mathcal G^0$ and expand every meta-good to its constituent items.
This gives a finite lottery over allocations of $M$.

For ex-ante WEF, the residual-chore marginal is the WEF DSE allocation.
Fix an observer $i$ and a meta-good $G$.
If $i\in P_G$, then \eqref{eq:bobw-safe-marginal} assigns $G$ to every $j\in P_G$ with probability $w_j/w(P_G)$; hence the normalized expected value of $G$ is $v_i(G)/w(P_G)$ both for $i$ and for every target in $P_G$, and is zero for targets outside $P_G$.
If $i\notin P_G$, its own contribution is zero and every target contribution is nonpositive.
The light-only blocks satisfy the same comparison by \cref{lem:bobw-subjective-goods}; for an observer outside $L$, they are nonpositive to every target.
Neutral blocks are harmless.
Additivity therefore proves ex-ante WEF.

It remains to verify the outcome-wise guarantee.
Fix an outcome and any two agents $i,j\in N$.
Let $\mathcal Q_i,\mathcal Q_j$ be their families of owned nonneutral meta-goods, let $Q_i=\bigcup_{G\in\mathcal Q_i}G$, and set $X_i=Y_i\cup Q_i$.
Every agent values $Q_i$ nonnegatively.
From $i$'s perspective, put
$$
 h=\sum_{G\in\mathcal Q_j:\,v_i(G)>0}v_i(G),
 \qquad t=v_i(Q_j)\le h.
$$
By Lemma~\ref{lem:construction}, every residual chore $e$ satisfies $v_i(e)<-h$.

If $q_i=q_j=0$, both agents are in $L$ and neither owns a block selected by a safe-owner kernel.
The remaining comparison follows from \cref{lem:bobw-subjective-goods}; when its target-block witness is $G$, use $g_i(G)$ from \eqref{eq:bobw-witness} to obtain one original good.
If $q_i=0<q_j$, then $v_i(Q_i)\ge0$, while one residual chore of $j$, together with all blocks of $j$, has negative value to $i$ by Lemma~\ref{lem:construction}; hence $i$ does not envy $j$.

Suppose $q_i>0=q_j$.
If $i\in L$, then $q_i=1$, and transferring its unique residual chore leaves $i$ with a nonnegative block bundle while the target bundle plus that chore has negative value to $i$.
If $i\in H$, every block owned by $j$ is nonpositive to $i$.
Applying \eqref{eq:bobw-ordinal-chore} with $\widehat v_i=v_i$ and transferring $i$'s worst residual chore proves $\WEFoneT$.

Finally, suppose $q_i,q_j>0$.
If $h>0$ and $i\in H$, every block counted in $h$ is assigned through a safe-owner kernel, so $h(q_i-1)/w_i\le hq_j/w_j$; the same inequality is immediate when $h=0$ or $i\in L$.
Set $v'_i(e)=v_i(e)+h$.
These values are negative and have the same order as $v_i$.

Let $e_i$ be $i$'s worst residual chore.
Applying \eqref{eq:bobw-ordinal-chore} with $\widehat v_i=v'_i$ and using the preceding count inequality yields
$$
 \frac{v_i(Y_i\setminus\{e_i\})}{w_i}
 \ge \frac{h+v_i(Y_j)+v_i(e_i)}{w_j}.
$$
Since $v_i(Q_i)\ge0$ and $t\le h$, this is exactly $v_i(X_i\setminus\{e_i\})/w_i\ge v_i(X_j\cup\{e_i\})/w_j$.
Thus one original own chore witnesses $\WEFoneT$ in every case.

To compress the support, represent every realized allocation by its $n\times m$ incidence vector.
These vectors lie in the affine space in which each item column sums to one, of dimension at most $m(n-1)$.
By Carath\'eodory's theorem, the expected incidence vector is a convex combination of at most $m(n-1)+1$ realized support allocations.
This keeps all itemwise marginals, hence preserves ex-ante WEF for additive valuations, and every retained allocation remains $\WEFoneT$.
This completes the proof of Theorem~\ref{thm:bobw}.

\section{Conclusion and Open Problem}\label{sec:conclusion}
In this work, we establish that a $\mathrm{WEF1}$ allocation always exists for additive mixed manna under arbitrary positive entitlements and can be computed in polynomial time. 
Furthermore, we achieve a best-of-both-worlds guarantee via a finite lottery that is ex-ante $\mathrm{WEF}$ and ex-post $\mathrm{WEF1T}$. 
Finally, we clarify the boundary between fairness and efficiency. 
While $\mathrm{WEF1}$ and $\mathrm{fPO}$ are generally incompatible, we show that every instance admits an allocation that simultaneously satisfies $\mathrm{WEF1T}$ and $\mathrm{fPO}$. 
Since $\mathrm{fPO}$ is strictly stronger than $\mathrm{PO}$, the impossibility result does not rule out the existence of allocations that are both $\mathrm{WEF1}$ and $\mathrm{PO}$, leaving this fundamental question open even for equal entitlements.

\section*{Declaration of the Use of AI Tools}
The proof ideas in this paper were developed independently by the authors. 
OpenAI's GPT-5.6 Sol was employed as a research assistant to assist in exploring and structuring the proof constructions. 
The authors remain fully responsible for the content, correctness, and integrity of all results in this paper.

\bibliographystyle{alpha}
\bibliography{main}
\newpage \appendix \section{Missing Proofs in Section~\ref{sec:picking}}\label{appendix: 3.2}
\begin{proof}[Proof of Lemma~\ref{lem:forward-metagoods}]
Let
$$
\mathcal G^+=\{G\in\mathcal G:\max_{h\in N'}\{v_h(G)\}>0\},
\qquad
\mathcal G^0=\mathcal G\setminus\mathcal G^+.
$$
For the allocation of $\mathcal G^+$, use the positive-part values
$$
\widehat v_i(G)=\max\{v_i(G),0\}
$$
and run \cref{alg:forward-picking}.
For any family $\mathcal H$ of meta-goods, write
$$
v_i(\mathcal H)=\sum_{G\in\mathcal H}v_i(G),
\qquad
\widehat v_i(\mathcal H)=\sum_{G\in\mathcal H}\widehat v_i(G).
$$
For each agent $h$, let
$$
X_h=\bigcup_{G\in Y_h}G
$$
be the corresponding bundle of original items.
Whenever an agent is selected, some remaining meta-good has positive value to her, so her favorite remaining meta-good is also positive.
Hence every assigned meta-good is strictly positive to its owner under both $\widehat v$ and $v$.

Fix distinct agents $i,j\in N'$.
List the meta-goods assigned to $i$ in chronological order as $x_1,\ldots,x_k$.
Among the meta-goods assigned to $j$, list those that are positive to $i$ as $z_1,\ldots,z_\ell$, again in chronological order.
We prove
\begin{equation}\label{eq:forward-target}
\frac{\widehat v_i(Y_i)}{w_i}
\ge
\frac{\widehat v_i(Y_j\setminus\{z_1\})}{w_j}
\end{equation}
when $\ell\ge1$; when $\ell=0$, the ordinary comparison holds because the right-hand side is zero and the left-hand side is nonnegative.
If $\ell=1$, the right-hand side of \eqref{eq:forward-target} is zero, so the claim is immediate.
Assume henceforth that $\ell\ge2$.

Represent the meta-goods of $i$ by consecutive intervals of length $1/w_i$ and define $\rho:(0,k/w_i]\to\R_{\ge0}$ by
$$
\rho(\alpha)=\widehat v_i(x_r)
\quad\text{for }
\alpha\in\left(\frac{r-1}{w_i},\frac{r}{w_i}\right],
\quad r\in[k].
$$
Similarly, define $\rho':(0,\ell/w_j]\to\R_{\ge0}$ by
$$
\rho'(\beta)=\widehat v_i(z_s)
\quad\text{for }
\beta\in\left(\frac{s-1}{w_j},\frac{s}{w_j}\right],
\quad s\in[\ell].
$$
This gives
\begin{equation}\label{eq:forward-areas}
\frac{\widehat v_i(Y_i)}{w_i}
=\int_0^{k/w_i}\rho(\alpha)\dd\alpha,
\qquad
\frac{\widehat v_i(Y_j\setminus\{z_1\})}{w_j}
=\int_{1/w_j}^{\ell/w_j}\rho'(\beta)\dd\beta.
\end{equation}
Meta-goods of $Y_j$ not listed among $z_1,\ldots,z_\ell$ have zero $\widehat v_i$-value and therefore do not appear in the second integral.

\emph{Horizontal lengths.} Immediately before $j$ receives $z_s$, where $s\ge2$, agent $i$ is active because $z_s$ is still available and is positive to $i$.
At this moment, $j$ has already received at least $s-1$ meta-goods.
Since $j$ is selected with minimum normalized pick count among the active agents, agent $i$ has normalized count at least $(s-1)/w_j$.
Taking $s=\ell$ and then using the final count of $i$ gives
\begin{equation}\label{eq:forward-length}
\frac{\ell-1}{w_j}\le\frac{k}{w_i}.
\end{equation}

\emph{Heights.} Take any $\beta\in(1/w_j,\ell/w_j]$ that is not an endpoint of one of the intervals above, and let $z_s$ be the meta-good represented at $\beta$.
Put $\alpha=\beta-1/w_j$.
Immediately before $j$ receives $z_s$, agent $i$ has normalized count at least $(s-1)/w_j\ge\alpha$.
Hence the meta-good of $i$ represented at position $\alpha$ was selected before $z_s$.
Since $z_s$ was still available at that earlier selection, the favorite-item rule gives
\begin{equation}\label{eq:forward-height}
\rho(\beta-1/w_j)\ge\rho'(\beta).
\end{equation}
The finitely many interval endpoints do not affect the integrals.

Combining \eqref{eq:forward-areas}--\eqref{eq:forward-height}, we obtain
\begin{align*}
\frac{\widehat v_i(Y_j\setminus\{z_1\})}{w_j}
&=\int_{1/w_j}^{\ell/w_j}\rho'(\beta)\dd\beta\\
&\le\int_0^{(\ell-1)/w_j}\rho(\alpha)\dd\alpha
\le\int_0^{k/w_i}\rho(\alpha)\dd\alpha
=\frac{\widehat v_i(Y_i)}{w_i}.
\end{align*}
Thus \eqref{eq:forward-target} holds.

We next return from $\widehat v$ to the original values $v$.
Every meta-good owned by agent $i$ is positive to $i$, so the value of $i$'s own bundle is unchanged.
For every meta-good in $j$'s bundle, $v_i(G)\le\widehat v_i(G)$, so the right-hand side can only decrease.
Moreover, $\widehat v_i(z_1)>0$ implies $v_i(z_1)>0$.
We therefore have the meta-level inequality
$$
\frac{v_i(Y_i)}{w_i}
\ge
\frac{v_i(Y_j\setminus\{z_1\})}{w_j}.
$$
Because $z_1$ is a meta-good and $v_i(z_1)>0$, choose an original item $g\in z_1$ such that $v_i(g)>0$ and $v_i(z_1\setminus\{g\})\le0$.
After expanding the meta-goods,
$$
v_i(X_j\setminus\{g\})
=v_i(Y_j\setminus\{z_1\})+v_i(z_1\setminus\{g\})
\le v_i(Y_j\setminus\{z_1\}).
$$
Thus the same comparison is certified by deleting the single original item $g$.
Since $i$ and $j$ were arbitrary, the expanded allocation of $\mathcal G^+$ is $\WEFone$.
Every agent values her own bundle nonnegatively because every assigned meta-good is positive to its owner.

Finally, each $G\in\mathcal G^0$ is nonpositive to every agent and has value zero to at least one agent.
Assign it intact to a zero valuer and add it to that agent's original bundle.
That agent's own value is unchanged, while her bundle becomes weakly less attractive to every other agent.
Hence all existing comparisons are preserved.
Repeating this operation allocates all meta-goods and proves the lemma.
\end{proof}

\section{Proof of Lemma~\ref{thm:wef1-fpo-impossibility}}
\label{app:wefxy-fpo-impossibility}
\begin{proof}
Fix $x,y\in[0,1]$ with $x+y<2$, and choose an integer $r\ge2$ such that $x+y<2-1/r$.
Consider three agents with equal entitlements $(w_1,w_2,w_3)=(1/3,1/3,1/3)$ and six items $e_1,\ldots,e_6$.
Their valuations are
\[
\begin{array}{c|rrrrrr}
 &e_1&e_2&e_3&e_4&e_5&e_6\\ \hline
1&{10r^2}&2&{5r}-1&-({10r^2}+1)&-{5r}&-2\\
2&{10r^2}&1&1&-{10r^2}&-({5r}+1)&-({5r}+1)\\
3&{10r^2}-1&1&{5r}&-{10r^2}&-({5r}+1)&-1
\end{array}.
\]
Thus $e_1,e_2,e_3$ are goods and $e_4,e_5,e_6$ are chores for all agents.
For a word $a_1\cdots a_6\in\{0,1,2\}^6$, let $X(a_1\cdots a_6)$ denote the allocation that assigns $e_t$ to agent $a_t+1$.

We use the standard positive-support characterization: an allocation $X$ is fPO if and only if there is $\lambda\in\mathbb R_{>0}^3$ such that the owner of every item $e$ maximizes $\lambda_i v_i(e)$.
Indeed, such an allocation maximizes $\sum_i\lambda_i v_i(X_i)$ over all fractional allocations, and any fractional Pareto improvement would strictly increase this welfare.

We next record the fPO allocations of this instance.
Normalize $\lambda_1=1$ and write $a=\lambda_2/\lambda_1$ and $b=\lambda_3/\lambda_1$.
For each owner word, the characterization above gives six score inequalities in $a,b$.
Checking the $3^6$ owner words gives precisely the following $43$ feasible words; the middle column specifies an ordered pair used below.
Conversely, every word in the table satisfies its score inequalities for some $a,b>0$ and is fPO.
\[
\begin{array}{c|c|l}
\text{group}&(i,j)&\text{owner words}\\ \hline
\mathrm A&(2,1)&
 \begin{gathered}
 \mathtt{000111},\mathtt{000112},\mathtt{000122},\mathtt{002102},\\
 \mathtt{002111},\mathtt{002112},\mathtt{002122}
 \end{gathered}\\[2mm]
\mathrm B&(3,1)&
 \mathtt{000212},\mathtt{000222},\mathtt{002202},
 \mathtt{002212},\mathtt{002222}\\[1mm]
\mathrm C&(3,2)&
 \begin{gathered}
 \mathtt{100222},\mathtt{102202},\mathtt{102222},\mathtt{110222},\\
 \mathtt{111202},\mathtt{111222},\mathtt{112202},\mathtt{112222}
 \end{gathered}\\[2mm]
\mathrm D&(1,2)&
 \begin{gathered}
 \mathtt{102000},\mathtt{102002},\mathtt{111000},\mathtt{111002},\\
 \mathtt{112000},\mathtt{112002},\mathtt{122000},\mathtt{122002}
 \end{gathered}\\[2mm]
\mathrm E&(1,3)&\mathtt{102102}\\[1mm]
\mathrm F&(1,3)&
 \mathtt{202000},\mathtt{202002},\mathtt{222000},\mathtt{222002}\\[1mm]
\mathrm G&(2,3)&
 \begin{gathered}
 \mathtt{202100},\mathtt{202102},\mathtt{202110},\mathtt{202111},\mathtt{202112},\\
 \mathtt{222100},\mathtt{222102},\mathtt{222110},\mathtt{222111},\mathtt{222112}
 \end{gathered}
\end{array}
\]

We show that every allocation in this table fails $\mathrm{WEF}(x,y)$.
Fix one such allocation $X$ and use the ordered pair $(i,j)$ associated with its group.
Write $U=v_i(X_i)$, $T=v_i(X_j)$, and $\Delta=T-U$.
Since entitlements are equal, their common denominator cancels from every WEF comparison.
Define
\[
h_{ij}(X)=
\max\left(
 \{0\}\cup\{-v_i(e):e\in X_i,\ v_i(e)<0\}
 \cup\{v_i(e):e\in X_j,\ v_i(e)>0\}
\right).
\]
For the seven groups, direct addition gives the following certificate.
The last column shows that $\Delta\ge(2-1/r)h_{ij}(X)$.
\[
\begin{array}{c|c|c|c}
\text{group}&\text{lower bound on }\Delta&h_{ij}(X)&
r\Delta-(2r-1)h_{ij}(X)\ge\\ \hline
\mathrm A&2{10r^2}-{5r}&{10r^2}&5r^2\\
\mathrm B&2{10r^2}-2{5r}&{10r^2}&0\\
\mathrm C&2{10r^2}-{5r}&{10r^2}&5r^2\\
\mathrm D&2{10r^2}+{5r}-1&{10r^2}+1&15r^2-3r+1\\
\mathrm E&2{5r}-5&{5r}&0\\
\mathrm F&2{10r^2}+2{5r}-4&{10r^2}+1&20r^2-6r+1\\
\mathrm G&2{10r^2}-{5r}&{10r^2}&5r^2
\end{array}.
\]
Every displayed lower bound on $\Delta$ is positive.
Thus $i$ envies $j$ without any adjustment.
Moving an eligible good from $j$ to $i$, or an eligible chore from $i$ to $j$, can reduce this gap by at most $(x+y)h_{ij}(X)$.
Since $x+y<2-1/r$, this amount is strictly smaller than $\Delta$.
Thus none of the three alternatives in the definition of $\mathrm{WEF}(x,y)$ holds for $(i,j)$.
Hence every fPO allocation fails $\mathrm{WEF}(x,y)$, proving the lemma.
\end{proof}

\end{document}